\documentclass[twoside]{aiml}

\usepackage[utf8]{inputenc}
\usepackage{latexsym,amssymb,amsmath,amsfonts}
\usepackage{color}
\usepackage{thmtools,thm-restate}
\usepackage{colonequals}
\usepackage{cancel}
\usepackage[dvipsnames]{xcolor}
\usepackage{tikz}
\usepackage{tikzit}
\usepackage{tikzsymbols}
\usepackage{float}
\usepackage{cite}
\usepackage{paralist}
\usepackage{mathtools}
\usepackage{url}
\usepackage{stmaryrd}

\usepackage{aimlmacro}

\usepackage{hyperref}

\newcommand{\imp}{\rightarrow}

\newcommand{\et}{\wedge}

\newcommand{\Et}{\bigwedge}

\renewcommand{\phi}{\varphi}

\newcommand{\inter}{\cap}

\newcommand{\cce}{\coloncolonequals}
\newcommand{\ce}{\colonequals}
\newcommand{\SFive}{{\sf S5}}
\newcommand{\KBFour}{{\sf KB4}}

\newcommand{\M}{\widehat{K}}

\newcommand{\localK}{\mathcal L^-}
\newcommand{\glocalK}{\mathcal L^+}

\newcommand{\fctdef}{\bowtie}
\newcommand{\ltdef}{\bowtie}
\newcommand{\nfctdef}{\not\bowtie}
\newcommand{\nltdef}{\not\bowtie}
\newcommand{\fctsat}{\vDash}
\newcommand{\nfctsat}{\nvDash}

\newcommand{\model}{\mathcal M}
\newcommand{\C}{\mathcal C}
\newcommand{\B}{\mathcal B}
\newcommand{\VV}{\mathcal V}
\newcommand{\FF}{\mathcal{F}}
\newcommand{\weg}[1]{}

\newcommand{\lbr}{[\![}
\newcommand{\rbr}{]\!]}
\newcommand{\I}[1]{\lbr #1 \rbr}

\newcommand{\sstar}{\mathsf{star}}

\newcommand{\loc}{\mathsf{par}}
\newcommand{\sta}{\mathsf{st}}

\newcommand{\lt}{\mathfrak{\rev{T}}}
\newcommand{\ltf}[1]{\lt_{{#1}}}
\newcommand{\treeLab}{\mathfrak{X}} 
\newcommand{\Tree}{\mathcal{T}}
\newcommand{\depth}{\mathop{\mathsf{Depth}}}
\newcommand{\bisim}{\leftrightarroweq}
\newcommand{\nbisim}{\not\bisim}

\newcommand{\for}[2]{\varphi^{\bowtie}_{#1}}
\newcommand{\forrk}[2]{#2^{\bowtie}_{#1}}

\newcommand{\Z}{{\mathcal Z}}

\tikzstyle{a dot}=[fill=green, shape=circle, inner sep=1]
\tikzstyle{b dot}=[fill=yellow, shape=circle,  inner sep=1]
\tikzstyle{c dot}=[fill=pink, shape=circle, inner sep=1]
\tikzstyle{d dot}=[fill=blue!20, shape=circle, inner sep=1]
\tikzstyle{e dot}=[fill=orange!20, shape=circle,  inner sep=1]
\tikzstyle{abig dot}=[minimum size=5ex,fill=green, shape=circle, inner sep=1]
\tikzstyle{bbig dot}=[minimum size=5ex,fill=yellow, shape=circle,  inner sep=1]
\tikzstyle{cbig dot}=[minimum size=5ex,fill=pink, shape=circle, inner sep=1]
\tikzstyle{dbig dot}=[minimum size=5ex,fill=blue!20, shape=circle, inner sep=1]
\tikzstyle{ebig dot}=[minimum size=5ex,fill=orange!20, shape=circle,  inner sep=1]
\tikzstyle{Kstate} = [fill = gray!20, shape=rectangle, rounded corners=5pt]

\def\lastname{B\'\i{}lkov\'a, van Ditmarsch, Kuznets, and Randrianomentsoa}

\newcommand{\rev}[1]{{#1}}

\begin{document}

\begin{frontmatter}

  \title{Bisimulation for Impure Simplicial Complexes}

  \author{Marta B{\'\i{}}lkov{\'a}}\footnotemark[1]%
  \qquad\footnotetext[1]{Work of Marta B\'\i{}lkov\'a on this paper was supported by the grant no. 22-23022L CELIA of the Czech Science Foundation.}%
  \and%
  \qquad%
  \author{Hans van Ditmarsch}
  \address{Czech Acad Sci, Inst Comp Sci\qquad\qquad Un.~of Toulouse, CNRS, IRIT, France}
  \author{\qquad\!\!\!\!\! Roman Kuznets}\footnotemark[2]\footnotetext[2]{This research was funded in whole or in part by the Austrian Science Fund (FWF) project ByzDEL [\href{https://doi.org/10.55776/P33600}{10.55776/P33600}].} 
 \quad\!\!\and\quad\,\,
  \author{Rojo Randrianomentsoa}\footnotemark[2]
  \address{TU Wien, Austria}

	\begin{abstract} 
	As~an alternative to Kripke models, simplicial complexes are a versatile semantic primitive on which to interpret epistemic logic. 
	Given a set of vertices, a simplicial complex is a downward closed set of subsets, called simplexes, of the vertex set. 
	A~maximal simplex is called a facet. Impure simplicial complexes represent that some agents~(processes) are dead. It~is known that impure simplicial complexes categorically correspond to so-called partial epistemic~(Kripke) models. 
	In~this contribution, we define a notion of bisimulation to compare impure simplicial complexes and show that it has the Hennessy--Milner property. 
	These results are for a logical language including atoms that express whether agents are alive or dead. 
	Without these atoms \rev{no reasonable standard notion of  bisimulation exists}, as we amply justify by counterexamples, \rev{because 
	such a restricted language is  insufficiently expressive}.
	\end{abstract}

\end{frontmatter}

\section{Introduction}
\label{intro}

\rev{Impure simplicial complexes provide semantics for multiagent epistemic logic for distributed systems with crash failures. 
In~this paper, we consider a three-valued semantics with the third value `undefined' used, e.g.,~for propositional atoms and knowledge statements pertaining to  crashed agents. 
We~show that the Hennessy--Milner property fails for the standard notion of bisimulation~and standard epistemic language. 
To~ameliorate the situation, we extend the language with global atoms representing whether a given agent is alive or crashed and prove the Hennessy--Milner property for this extended language. 
The~standard proof of this property, which relies on the symmetry between the two boolean truth values, does not work for our three-valued logic. 
To~adapt~it, we employ  a tailor-made notion of a life tree, which encodes definability and enables us to define a localized  translation from undefined to true formulas.}
\smallskip

\noindent
\textbf{Survey of the literature}
Simplicial complexes are well-known in combinatorial topology. 
There~have been promising and exciting recent connections between combinatorial topology and epistemic logic~\cite{goubaultetal_postdali:2021,vDitKuz22arXiv,hvdetal.simpl:2022,goubaultetal:2021,GoubaultKLR23,CachinLS23,HaRoRo,GoubaultKL24,cdrv:2023,CastanedavDKMS24}.

Combinatorial topology has been used in distributed computing to model concurrency and asynchrony since~\cite{BiranMZ90,FischerLP85,luoietal:1987}. 
Higher-dimensional topological properties~\cite{HS99,herlihyetal:2013} allow for an epistemic representation. 
The~basic structure in combinatorial topology is the \emph{simplicial complex}, a downward closed collection of subsets called \emph{simplexes} of a set of \emph{vertices}. 
Geometric manipulations such as subdivision have natural combinatorial counterparts. 

Epistemic logic investigates knowledge and belief, and change thereof, in multiagent systems. 
A~foundational study is~\cite{hintikka:1962}. 
Knowledge change was extensively modeled in temporal epistemic logics~\cite{halpernmoses:1990,Pnueli77,dixonetal.handbook:2015} and in dynamic epistemic logics~\cite{baltagetal:1998,hvdetal.del:2007}.

An~epistemic logic interpreted on \emph{pure} simplicial complexes was proposed in~\cite{goubaultetal_postdali:2021}. 
It~shows a categorical correspondence between Kripke models and simplicial complexes, and based on that, the resulting logic is \rev{multiagent}~$\SFive$ augmented with the locality axiom $K_a p_a \vee K_a \neg p_a$ stating that all agents know their local state. 
Action models~\cite{baltagetal:1998} are used to model distributed computing tasks. \looseness=-1

In~\rev{\emph{pure complexes}} and their temporal developments, all processes remain active (are~alive). 
They~describe \emph{asynchronous} message passing. In impure complexes some processes may have crashed (are~dead). 
They~can be used to describe \emph{synchronous} message passing (with~timeouts)~\cite{herlihyetal:2013}. 
\rev{\emph{Impure complexes}} correspond to Kripke models with partial equivalence relations (symmetric~and transitive relations). 
Epistemic logics interpreted on impure simplicial complexes were proposed in~\cite{vDitKuz22arXiv,goubaultetal:2021}. 
In~\cite{goubaultetal:2021}, a two-valued semantics is proposed and the authors axiomatize the logic as \rev{multiagent}~$\KBFour$ where, if process~$a$ is dead, then $K_a \bot$~is true. In~\cite{vDitKuz22arXiv}, a three-valued modal logical semantics  is proposed for the same language where the third value stands for `undefined', e.g.,~dead processes cannot know or be ignorant of any proposition, nor can  live processes know or be ignorant of factual propositions involving processes they know to be dead. 
This~logic was axiomatized in~\cite{RandrianomentsoaDK23} in a version of \rev{multiagent}~$\SFive$ called~$\SFive^{\bowtie}$. 
Its~notion of knowledge relates to ``belief as defeasible knowledge'' of~\cite{MosesS93}. 
For example, $K_a \varphi \imp \varphi$ is valid in the sense that, if $K_a \varphi$~is true, then $\varphi$~is not false (but~may be~undefined). 
The~three- and two-valued approaches are compared in~\cite{HaRoRo}.

Subsequent developments of simplicial epistemics include generalizations from individual knowledge to distributed knowledge~\cite{GoubaultKLR23} and from simplicial complexes to (semi-)simplicial sets~\cite{GoubaultKLR23,CachinLS23,GoubaultKL24}. 
Dynamics of complexes were investigated in~\cite{cdrv:2023}.

The~propositional base of our three-valued semantics is known as Paraconsistent Weak Kleene logic~(PWK)~\cite{Kleene1952-KLEITM,Hallden1949-HALTLO-7,SzmucForthcoming-SZMAEI,CiuniCarrara2019,BonzioGPP17}. 
Bisimulations and limits of expressivity (Hennessy--Milner property) for many-valued modal logics based on algebra-valued Kripke frames, or more generally coalgebras, were considered in~\cite{Marti2014-MARAHP-6,BilkovaDostal2016}, restricting to algebras in question being residuated lattices or even finite MTL~chains (of~which the PWK~three-valued matrix is~neither). 
\smallskip

\noindent
\textbf{Motivating example}
Consider three agents~$a$, $b$,~and~$c$ with \emph{local atoms}~$p_a$, $p_b$,~and~$p_c$ respectively describing their local state. 
The~value for~$a$ is~$1$ when $p_a$~is true or~$0$ when false,  and similarly for~$b$~and~$c$. 
In~combinatorial topology, such information can be represented in a simplicial complex. 
In~simplicial model~$\C'$ of Fig.~\ref{fig.motivating}, \rev{every agent knows (i.e.,~there is mutual knowledge) that}  the values of~$a$~and~$c$ are~$1$ and that of~$b$ is~$0$. 
In~simplicial model~$\C$, on the other hand, agent~$a$ is uncertain whether agent~$c$ is still alive. 
This~uncertainty is represented by  edge~$X$ and triangle~$Y$ intersecting in the vertex \rev{(labeled~with)}~$1_a$.

\begin{figure}[h]
\centering
\scalebox{.8}{
\begin{tikzpicture}
\fill[fill=gray!25!white] (2,0) -- (4,0) -- (3,1.51) -- cycle;
\node (c) at (-1,0) {$\C:$};
\node[b dot] (b1) at (0,0) {$0_b$};
\node[b dot] (b0) at (4,0) {$0_b$};
\node[c dot] (c1) at (3,1.51) {$1_c$};
\node[a dot] (a0) at (2,0) {$1_a$};
\node (f1) at (3,.65) {$Y$};
\draw[-] (b1) -- node[above] {$X$} (a0);
\draw[-] (a0) -- (b0);
\draw[-] (b0) -- (c1);
\draw[-] (a0) -- (c1);
\end{tikzpicture}
\qquad\qquad
\begin{tikzpicture}
\fill[fill=gray!25!white] (2,0) -- (4,0) -- (3,1.51) -- cycle;
\node (c) at (1,0) {$\C':$};
\node[b dot] (b0) at (4,0) {$0_b$};
\node[c dot] (c1) at (3,1.51) {$1_c$};
\node[a dot] (a0) at (2,0) {$1_a$};
\node (f1) at (3,.65) {$Y'$};

\draw[-] (a0) -- (b0);
\draw[-] (b0) -- (c1);
\draw[-] (a0) -- (c1);

\end{tikzpicture}
}

\bigskip

\scalebox{.8}{
\begin{tikzpicture}
\node (c) at (-1,0) {$\model:$};
\node[Kstate] (010) at (.5,0) {$1_a0_b$};
\node[Kstate] (001) at (3.5,0) {$1_a0_b1_c$};
\draw[ - ] (010) -- node[above] {$a$} (001);
\end{tikzpicture}
\qquad\qquad\qquad
\begin{tikzpicture}
\node (c) at (-1,0) {$\model':$};
\node[Kstate] (010) at (.5,0) {$1_a0_b1_c$};
\end{tikzpicture}
}
\caption{Simplicial models and corresponding partial epistemic models}
\label{fig.motivating}
\end{figure}
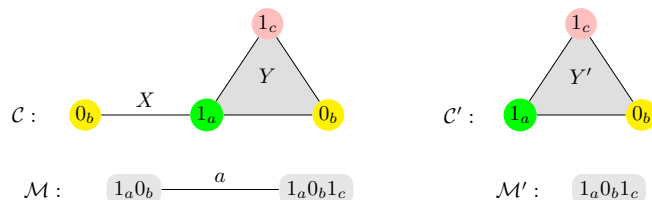

Model~$\C$ encodes that $a$~is uncertain whether agent~$c$ is alive but knows, e.g.,~that the value of~$b$ is~$0$ as this is the case in both edge~$X$ and triangle~$Y$\rev{, which} intersect in the $a$-vertex. 
Model~$\C'$ encodes that $a$~knows that agent~$c$ is alive, and also that $a$~knows that the value of~$b$ is~$0$. 
In~our logical semantics, one cannot evaluate formulas referring to the knowledge or local \rev{atoms} of agents that are dead. 
These are then undefined. 
The~semantics is, therefore, three-valued: a formula can be true, false, or undefined. For example, in edge~$X$ of~$\C$, formulas~$p_c$ (the~atom) and $K_c p_a$~are undefined. 
However, also in~$X$, a~formula like~$\M_a p_c$ is defined (and~true): $a$~considers it possible (namely, in the indistinguishable triangle~$Y$) that the value of~$c$ is~$1$ (that $p_c$~is true). 

A~so-called \emph{global atom} like~$a$, $b$,~or~$c$ represents that agents~$a$, $b$,~and~$c$ respectively are alive. 
It~is elementary to distinguish  structures $(\C,Y)$ and $(\C',Y')$ this way: in~$(\C,Y)$ it is true that $a$~considers it possible that $c$~is dead, formalized as~$\M_a \neg c$, whereas in~$(\C',Y')$ this is false\rev{:
a}gent~$a$ there knows that $c$~is alive,~$K_a c$. 
A~corresponding notion of bisimilarity between structures should, therefore, also distinguish \rev{them. 
In}~this contribution, we will define such a notion and show its elementary properties.

In~the language without such global atoms, we cannot say that $a$~considers it possible that $c$~is dead. 
\rev{For~instance, a naive attempt to use~$K_c \top$ to represent global atom~$c$ fails for two separate reasons. 
Firstly, $\neg K_c \top$~is undefined rather than true when $c$~is actually dead; secondly, the boolean constant~$\top$, which is true in all structures, is itself not expressible without global atoms. 
Indeed, for  every formula without global atoms, there is a singleton structure where this formula is undefined~\cite{RandrianomentsoaDK23}.} 

\rev{More generally, there is no  way to distinguish $(\C,Y)$ from~$(\C',Y')$ without expanding the language.} 
We~show \rev{(Prop.~\ref{prop.llll})} that in the logical language with atoms~$p_a$, $p_b$,~and~$p_c$ and modalities~$K_a$, $K_b$,~and~$K_c$ only, $(\C,Y)$ and $(\C',Y')$ make the same formulas true, false, and undefined: they are \rev{\emph{modally equivalent}}. 
Therefore, \rev{$(\C,Y)\bisim^\ast(\C',Y')$ should hold for any reasonable} notion~\rev{$\bisim^\ast$ of \emph{bisimilarity}}. 
But~now we have a problem: assuming such a notion, as $X$~and~$Y$~intersect in~$a$, when making a \rev{bisimilar} \textbf{forth} step from~$Y$ to~$X$, we should be able to emulate that in~$\C'$ by moving from~$Y'$  somewhere\rev{. 
But~no modally equivalent match exists, making it impossible to define $\bisim^\ast$.}

All~these issues similarly exist in the corresponding local epistemic~(Kripke) models~$\model$~and~$\model'$ pictured below simplicial models~$\C$~and~$\C'$ in Fig.~\ref{fig.motivating}. 
In~place of equivalence relations of  epistemic models, \rev{\emph{partial epistemic models}} have partial equivalence relations (symmetric~and transitive relations). 
Model~$\C'$ corresponds to the singleton model~$\model'$, with reflexive access for~$a$, $b$,~and~$c$ and making~$p_a$~and~$p_b$~false and $p_c$~true there, whereas model~$\C$ corresponds to a two-world Kripke model~$\model$ where only $a$~cannot distinguish the two worlds and wherein the relation for~$c$ is restricted to the world where it is alive. 
In~fact, $\model$~is an equivalence class of two Kripke models, \rev{one where $p_c$~is assigned value true in the left world and another one where it is assigned false.  
Since $c$~is dead in that world, however, the assigned value of~$p_c$ there is moot:  $p_c$~is considered undefined, allowing us to conflate the two models.} 
\smallskip

\noindent
\textbf{Results in this contribution}
We~distinguish two epistemic languages to describe impure simplicial complexes: one without and one with global atoms. 
These languages are interpreted on simplicial models, where we also extend a known correspondence between those and partial epistemic models to the language with global atoms. 

For~the language without global atoms, we prove that even simple structures like~$(\C,Y)$ and $(\C',Y')$ in Fig.~\ref{fig.motivating} may have the same information content (be~modally equivalent), but \rev{lack} a corresponding notion of bisimulation where \textbf{back} and \textbf{forth} steps can be locally checked.\footnote{\rev{Note that our analysis here and in Example~\ref{ex.faraway} pertains only to the standard notion of bisimulation. The existence of weaker notions,  e.g., of $\rho$-bisimulations~\cite{GrootHansenKurz2020}, which do not depend on expressivity, remains open.}}
\rev{One}~can only employ ``brute~force'' here, taking properties of the whole model into account. 
Hence, the language without global atoms \rev{is} insufficiently expressive. 

For~the language with global atoms, we define bisimulation on simplicial models. 
By~way of an intermediate structure called a \emph{life tree}, \rev{which provides an efficient way to check definability,}  we then establish \rev{the} Hennessy--Milner characterization. 
In~our three-valued semantics, with values true, false, and undefined, 
it is non-standard to show that modal equivalence implies bisimilarity on finitary structures. 
In~the usual method for \rev{the} two-valued semantics, one can, w.l.o.g.,~assume that a formula distinguishing a point in one structure from a (finite)~set of 
points in the other structure is true in the former and not true~(false) in the latter. 
But~in \rev{the}  three-valued semantics, instead of being true, \rev{it can be undefined}. 
In~that case, life trees and the local geometry of the former structure are used to construct another distinguishing formula that is true.\looseness=-1

Finally, we translate our results to a three-valued Kripke semantics interpreted on partial epistemic models. 
We define a \emph{life bisimulation} between partial epistemic models, which can \rev{be} easily seen to be weaker than standard bisimulation. 
\rev{The} Hennessy--Milner property \rev{follows from that} for complexes \rev{via a categorical equivalence}. 
\smallskip

\noindent
\textbf{Overview}
Section~\ref{sec.complex} defines the two logical languages and simplicial models and provides the simplicial semantics to interpret these languages on them. 
Section~\ref{sec.local} explains the difficulties of defining bisimulation for the language with only local atoms. 
Section~\ref{sec.bisimulation} then defines bisimulation for the language with global atoms and establishes the Hennessy--Milner characterization. 
Section~\ref{sec.kripke} presents the results for bisimulation in the corresponding setting of Kripke models.\looseness=-1
 
\section{Epistemic Simplicial Semantics}
\label{sec.complex}

We~consider a finite set~$A$ of \emph{agents} (or~\emph{processes})~$a,b,\dots$ and a  set
$P = A \sqcup \bigsqcup_{a \in A} P_a$
of \emph{propositional \rev{atoms}} where sets~$P_a$ are countable and mutually disjoint sets of \emph{local \rev{atoms} for agent~$a$}, denoted~$p_a, q_a, p'_a, q'_a, \dots$ 
We overload the meaning of~$A$ as it also stands for the set of \emph{global \rev{atoms}}. 
\rev{Atoms}~$a \in A$ represent `agent~$a$ is alive.' 
We define two languages:  the \emph{local language} with only local \rev{atoms} and the \emph{glocal \rev{language}}  with both local and global \rev{atoms}.

\begin{definition}
	Let  $a \in A$ and $p_a \in P_a$. \emph{Language~$\glocalK$}  is defined by 
	\[
	\varphi 
	\cce
	a \mid p_a \mid \neg\varphi \mid (\varphi\land\varphi) \mid \M_a \varphi.
	\]  
\emph{Language~$\localK$} is the fragment of~$\glocalK$ without global atoms~$a$. Boolean  connectives are defined in a standard way; $K^{\mathstrut}_a \varphi \ce \neg \M_a \neg \varphi$. \rev{In $\glocalK$, we can additionally define boolean constants  $\top \ce a \lor \neg a$ for some fixed $a \in A$ and $\bot \ce \neg \top$.}
\end{definition}

\begin{definition}
	A \emph{simplicial model}~$\C$ is a triple~$(C,\chi,\ell)$ such that:
	\begin{compactitem}
		\item 
		A (\emph{simplicial}) \emph{complex} $C\ne \varnothing$ is a collection of  \emph{simplexes} such that
		\begin{compactenum}
		 \item
		 every $X \in C$ is a non-empty finite subset of a given set~$\VV$  of \emph{vertices};
		 \item for each simplex $X \in C$, if  $\varnothing \ne Y \subseteq X$, then $Y \in C$;
		 \item $\{v\} \in C$  for each vertex $v \in \VV$.
		\end{compactenum}
		\item 
		A \emph{chromatic function} $\chi \colon \VV \to A$ is a map from vertices to agents such that
for each $X \in C$ and $v,u \in X$, if  $\chi(v) = \chi(u)$ , then $v = u$.
		\item 
		A \emph{valuation function} $\ell \colon \VV \to 2^{P{\setminus}A}$ is a map from vertices to sets of local \rev{atoms} such that $\ell(v) \subseteq P_a$ for each $v \in \VV$ with $\chi(v) = a$.
	\end{compactitem}
			For each simplex $X \in C$, we define $\chi(X) \ce \{\chi(v) \mid v \in X\}$ and $\ell(X)\ce\bigsqcup_{v \in X} \ell(v)$.  
			For arbitrary simplexes $X, Y \in C$, if $Y \subseteq X$, we say that $Y$~is a \emph{face} of~$X$. 
	Since each simplex is a face of itself, we use `simplex' and `face' interchangeably. Faces~$Y$~and~$Z$ are \emph{$a$-adjacent}  if{f} $a \in \chi(Y \cap Z)$.	  
	A face~$X$ is a \emph{facet} if{f} for any simplex $Y \in C$, if $Y \supseteq X$, then $Y = X$. 
The set of all facets is denoted~$\FF(C)$. 
	The \emph{dimension of a simplex}~$X$ is~$|X|-1$. 
	The \emph{dimension of a simplicial model~$\C$} is the largest dimension of its facets. 
	A simplicial model~$\C$ is \emph{pure} if{f} all its facets have dimension~$|A|-1$. 
	Otherwise, $\C$~is \emph{impure}. 
	A~\emph{pointed simplicial model} is a pair~$(\C,X)$ where $X \in \FF(C)$. We often omit `pointed'.
\end{definition} 

\begin{definition}
\label{lanKa.fct.defsat}
	Let $\C = (C,\chi,\ell)$ be a simplicial model. \emph{Definability relation~$\fctdef$} and \emph{satisfaction relation~$\fctsat$} for $X \in \FF(C)$ are defined recursively on $\varphi\in \glocalK$:\\
\strut\qquad	$\begin{array}{lcl}
		\C, X \fctdef a & &  \text{always}; \\
		\C, X \fctdef p_a & \text{if{f}} & a \in \chi(X); \\
		\C, X \fctdef \neg \varphi & \text{if{f}} & \C, X \fctdef \varphi; \\
		\C, X \fctdef \varphi\land\psi & \text{if{f}} & \C, X \fctdef \varphi \ \text{and} \ \C, X \fctdef \psi; \\
		\C,X \fctdef \M_a\varphi & \text{if{f}} & \C,Y \fctdef \varphi \ \text{for some} \ Y \in \FF(C) \ \text{with} \ a \in \chi(X \inter Y). \\[1.5ex]
	
		\C, X \fctsat a & \text{if{f}} &   a \in \chi(X); \\
		\C, X \fctsat p_a & \text{if{f}} &  \ p_a \in \ell(X); \\
		\C, X \fctsat \neg \varphi & \text{if{f}} & \C, X \fctdef \varphi \ \text{and} \ \C, X  \nfctsat \varphi; \\
		\C, X \fctsat \varphi\land\psi & \text{if{f}} & \C, X \fctsat \varphi \ \text{and} \ \C, X \fctsat \psi; \\
		\C,X \fctsat \M_a\varphi & \text{if{f}} & C,Y \fctsat \varphi \text{ for some } \ Y \in \FF(C)  \text{ with } a \in \chi(X\cap Y). 
	\end{array}$
\end{definition}

\begin{example} 
	\label{example.xxx}
	For simplicial model~$\C$ from Fig.~\ref{fig.motivating}:
	\begin{compactitem}
		\item \emph{Atoms and knowledge of dead agents.} Illustrating the novel aspects of the semantics,   $\C,X \nfctdef p_c$  since $c \notin \chi(X)=\{a,b\}$. Consequently, $\C,X \nfctdef \neg p_c$, $\C,X \nfctsat p_c$, and $\C,X \nfctsat \neg p_c$. For the same reason, $\C,X \nfctdef \M_c  p_a$. Thus, $\C,X \nfctdef \neg \M^{\mathstrut}_c p_a$, $\C,X \nfctsat \M_c  p_a$, and $\C,X \nfctsat \neg \M_c  p_a$. 
		
		\item \emph{Knowledge of a live agent concerning dead agents.} Although $\C,X\nfctdef p_c$, still $\C,X \fctsat \M_a p_c$ because $a \in \chi(X \inter Y)=\{a\}$ and $\C,Y \fctsat p_c$. 		More surprisingly, also $\C,X \fctsat K_a p_c$ because, given the two facets~$X$~and~$Y$ that agent~$a$ considers possible, as far as $a$~knows, $p_c$~is true. This knowledge is defeasible because $a$~may learn that the actual facet is~$X$ and not~$Y$, which she also considers possible.
	\end{compactitem}
\end{example}

\section{Bisimulation for Simplicial Models with Local Atoms?} \label{sec.local}

Before we define bisimulation for simplicial models in the language with global atoms, we explain the difficulties in finding such a notion for the language without. 
The~introductory section mentioned that pointed simplicial models~$(\C,Y)$~and $(\C',Y')$ from Fig.~\ref{fig.motivating} are modally equivalent in language~$\localK$. 
We~now establish it formally.

\begin{definition}
	\label{modalEquivalence}
	Pointed simplicial models $(\C,Y)$~and $(\C',Y')$ are   \emph{modally equivalent in \rev{a} language~$\mathcal{L}$}, written $(\C,Y)\equiv_{\mathcal{L}}(\C',Y')$, 	
	if{f} for \rev{each  $\varphi \in \mathcal{L}$}:\looseness=-1
		\begin{gather}
		\label{eq:modeq_def}
		\C,Y \fctdef \varphi \quad\Longleftrightarrow\quad \C',Y' \fctdef \varphi,
		\\
		\label{eq:modeq_true}
		\C,Y \fctsat \varphi \quad\Longleftrightarrow\quad \C',Y' \fctsat \varphi ,
		\\
		\label{eq:modeq_false}
		\C,Y \fctsat \neg \varphi\quad\Longleftrightarrow\quad \C',Y' \fctsat \neg \varphi.
	\end{gather}
\end{definition}
It is easy to show   that, in fact, \eqref{eq:modeq_true}~alone is sufficient:
\begin{lemma}[Criterion of modal equivalence]
	\label{simplifyModalEquivalence} 
If\/~\eqref{eq:modeq_true}~holds for all~$\varphi \in \mathcal{L}$ for pointed simplicial models\/ $(\C,Y)$~and\/ $(\C',Y')$, then\/ 
$(\C,Y)\equiv_{\mathcal{L}}(\C',Y')$.
\end{lemma}
\begin{proof}
	We~need  to establish~\eqref{eq:modeq_def}~and~\eqref{eq:modeq_false} for all~$\varphi \in \mathcal{L}$. 
	\eqref{eq:modeq_false}~for~$\varphi$ follows from~\eqref{eq:modeq_true} for~$\neg \varphi$. 
	To~show~\eqref{eq:modeq_def} for~$\varphi$, assume first that $\C,Y \fctdef \varphi$. 
	Then either $\C,Y \fctsat \varphi$~or $\C,Y \fctsat \neg \varphi$. 
	By~\eqref{eq:modeq_true}  for~$\varphi$~and~$\neg \varphi$,  either $\C',Y' \fctsat \varphi$~or $\C',Y' \fctsat \neg \varphi$. 
	In~either case, $\C',Y' \fctdef \varphi$. 
	We~have proved the left-to-right direction of~\eqref{eq:modeq_def}. 
	The~right-to-left direction is symmetric. 
\end{proof}

Lemma~\ref{simplifyModalEquivalence} seems to considerably simplify the definition of three-valued modal equivalence, as it now appears to be the same as in the two-valued semantics. 
In~practice, it is not much help in proofs, as $\C,Y\nfctsat\varphi$ may imply falsity as well as being undefined. 
It~is more practical to check both definability~\eqref{eq:modeq_def}  and truth~\eqref{eq:modeq_true}. 
(See,~e.g.,~the proof of Theorem~\ref{BisimImpliesEquiv}, where both~\eqref{eq:modeq_def}~and~\eqref{eq:modeq_true}  are needed in the induction step for~\eqref{eq:modeq_true} for~$\neg$.)

\begin{proposition} 
\label{prop.llll} 
$(\C,Y) \equiv_{\localK} (\C',Y')$ 
for\/  $(\C,Y)$~and\/ $(\C',Y')$ from Fig.~\ref{fig.motivating}.
\end{proposition}
\begin{proof}
	Let~us first prove  by induction on the construction of~\mbox{$\varphi \in \localK$} an auxiliary statement: 
	\begin{equation}
		\label{eq:aux}
		\C,X \fctdef \varphi 
		\quad\Longrightarrow\quad 
			(
			\C,X \fctsat \varphi 
			\Longleftrightarrow 
			\C,Y \fctsat \varphi
			).
	\end{equation}
	For~local \rev{atoms}, this is obvious from the construction of~$\C$.
	The~cases for~$\neg$~and~$\land$ are straightforward. 
	For~$\varphi = \M_a \psi$, the truth value is determined by the vertex~$1_a$, which is shared between~$X$~and~$Y$. 
	No~formula~$\M_c \psi$ is defined in~$X$. 
	It~remains to consider  
	$\varphi=\M_b \psi$ 
	such that 
	$\C,X \fctdef \M_b \psi$:
	\[
	\C,X \fctsat \M_b \psi 
	\quad\Longleftrightarrow\quad
	\C,X \fctsat \psi
	\quad\xLeftrightarrow{\text{IH}}\quad
	\C,Y \fctsat \psi
	\quad\Longleftrightarrow\quad
	\C,Y \fctsat \M_b\psi
	\]
	where the~IH can be applied to~$\psi$ because 
	$\C,X \fctdef \M_b \psi$ 
	implies 
	$\C,X \fctdef \psi$.
	
	By~Lemma~\ref{simplifyModalEquivalence}, it is sufficient to prove~\eqref{eq:modeq_true}. 
	(Note~that  
	\eqref{eq:modeq_def}~is trivial here since all formulas are defined in both~$Y$~and~$Y'$.) 
	We~use  induction on~\mbox{$\varphi \in \localK$}.
	For~local \rev{atoms}, \eqref{eq:modeq_true}~is obvious from the construction of~$\C$~and~$\C'$.
	The~cases for~$\neg$~and~$\land$ are straightforward.
	For~$\varphi = \M_i \psi$ with~$i \in \{b,c\}$, property~\eqref{eq:modeq_true} follows from the~IH and the fact that, for each of~$Y$~and~$Y'$, the only \rev{facet} $b$-/$c$-adjacent to it is itself.
	Finally, for~$\phi = \M_a \psi$,
	\begin{multline*}
		\C,Y \fctsat \M_a \psi
		\quad\Longleftrightarrow\quad
		\C,X \fctsat \psi \text{ or } \C,Y \fctsat \psi
		\quad\xLeftrightarrow{\eqref{eq:aux}}\quad
		\C,Y \fctsat \psi
		\quad\xLeftrightarrow{\text{IH}}\quad\\
		\quad\xLeftrightarrow{\text{IH}}\quad
		\C',Y' \fctsat \psi
		\quad\Longleftrightarrow\quad
		\C',Y' \fctsat \M_a \psi 
	\end{multline*}
where \eqref{eq:aux}~is only used when~$\C,X \fctsat \psi$ and, hence,~$\C,X \fctdef \psi$.
\end{proof}

A~natural notion of bisimulation for simplicial models, as adapted from  Kripke models, is the following:
\begin{definition}
	\label{def:bis0}
	A~\emph{bisimulation} between simplicial models $\C=(C,\chi,\ell)$~and $\C'=(C',\chi',\ell')$ is a non-empty binary relation~$\B\subseteq\FF(C)\times\FF(C')$ such that, whenever~$X \B X'$, the following conditions are fulfilled:
	\begin{compactdesc}
		\item[\textbf{Atoms}:] 
			$\chi(X)= \chi'(X')$~and $\ell(X) =\ell'(X')$.
		
		\item[\textbf{Forth}:] 
			For~each agent~$a \in A$, if~$Y \in \FF(C)$~and $a \in \chi(X \inter Y)$, then there exists~$Y' \in \FF(C')$ with~$a \in \chi'(X' \inter Y')$~and $Y \B Y'$.
	
		\item[\textbf{Back}:] 
			For~each agent~$a \in A$, if~$Y' \in \FF(C')$~and $a \in \chi'(X' \inter Y')$, then there exists~$Y \in \FF(C)$ with~$a \in \chi(X \inter Y)$~and $Y \B Y'$.
	\end{compactdesc}
Simplicial models are called \emph{bisimilar} if{f} there exists a bisimulation~$\B$ between~them.
Pointed simplicial models~$(\C,X)$~and $(\C',X')$ are \emph{bisimilar}, written $(\C,X)\bisim(\C',X')$, if{f} there exists a bisimulation~$\B$ between~$\C$~and~$\C'$ with~$X\B X'$.	
\end{definition}

A~proper notion of bisimulation should enjoy the Hennessy--Milner property: two pointed simplicial models are modally equivalent if{f} they are bisimilar. 
Since~$(C,Y)\equiv_{\localK}(\C',Y')$  by Prop.~\ref{prop.llll}, they are supposed to be bisimilar. 
The~\textbf{forth} requirement \rev{of Def.~\ref{def:bis0}} demands to find a facet of~$\C'$ bisimilar to  facet~$X$ of~$\C$. 
\rev{But~no facet of~$\C'$ can be bisimilar to~$X$  because none is modally equivalent to~$X$.%
\footnote{\rev{In fact, even if the semantics is defined on all simplexes rather than only on facets (see~\cite{HaRoRo}), even then there are no simplexes in~$\C'$ that are modally equivalent to~$X$ because~$\C, X \fctsat \neg p_b$ but~$\C, X \nfctsat \M_b p_c$.}}} 
\rev{Thus,~the only possibility to define bisimulation coinciding with modal equivalence would be to  rule out such troublesome adjacent \rev{facets} from the scope of  \textbf{forth}/\textbf{back}, similar to how  a disconnected component does not prevent establishing a bisimulation. 
The difficulty here is that $X$~is not disconnected. 
Morally,  $X$~need not be considered here in \textbf{forth} step from~$Y$ because strictly fewer formulas are defined in~$X$ than in~$Y$ (while~for all~$\varphi$, \mbox{$\C, X \fctdef \varphi$} implies $\C, Y \fctdef \varphi$~and $\C, X \fctsat \varphi$ implies $\C, Y \fctsat \varphi$). 
Unfortunately,  as we demonstrate in Example~\ref{ex.faraway}, no local condition exists that would enable us to make a determination whether to consider an adjacent \rev{facet} in \textbf{forth}/\textbf{back}  or not. 
More precisely, such a determination cannot be made  based exclusively on the values of~$\chi$~and~$\ell$ in all \rev{facets} reachable in (at most) $k$~consecutive steps from one \rev{facet} to an adjacent one, for any fixed~$k$.}

\begin{figure}[h]
\centering
\scalebox{.8}{
		\begin{tikzpicture}
			\node (Complex) at (-1.5,-.5) {$\C:$};
			\node[a dot] (a) at (0,0.3) {$0_a$};
			\node[c dot] (c) at (0,-1) {$1_{b}$};
			\node[b dot] (bl) at (-1,-2) {$0_{c}$};
			\node[b dot] (br) at (1,-2) {$0_{c}$};
			\node[d dot] (bbl) at (-1,-3) {$0_{d}$};
			\node[d dot] (bbr) at (1,-3) {$0_{d}$};
			\node[b dot] (bbbl) at (-1,-4) {$0_{c}$};
			\node[b dot] (bbbr) at (1,-4) {$0_{c}$};
			\node[d dot] (dl) at (-1,-5) {$0_d$};
			\node[d dot] (dr) at (1,-5) {$0_d$};
			\node[e dot] (el) at (-1,-6) {$1_e$};
			\node[e dot] (er) at (1,-6) {$1_e$};
			
			\draw[-] (a) -- node[left] {$X$} (c);
			\draw[-] (c) -- node[above,pos=.7] {$Y$} (bl);
			\draw[-] (c) -- (br);
			\draw[-] (bl) -- (bbl);
			\draw[-] (br) -- (bbr);
			\draw[dotted,thick] (bbl) -- (bbbl);
			\draw[dotted,thick] (bbr) -- (bbbr);
			\draw[-] (bbbl) -- (dl);
			\draw[-] (bbbr) -- (dr);
			\draw[-] (dl) -- (el);
			\draw[-] (dr) -- (er);
		\end{tikzpicture} \qquad\qquad
		\begin{tikzpicture}
			\node (Complex) at (-1.5,-.5) {$\C':$};
			\node[a dot] (a) at (0,0.3) {$0_a$};
			\node[c dot] (c) at (0,-1) {$1_{b}$};
			\node[b dot] (bl) at (-1,-2) {$0_{c}$};
			\node[b dot] (br) at (1,-2) {$0_{c}$};
			\node[d dot] (bbl) at (-1,-3) {$0_{d}$};
			\node[d dot] (bbr) at (1,-3) {$0_{d}$};
			\node[b dot] (bbbl) at (-1,-4) {$0_{c}$};
			\node[b dot] (bbbr) at (1,-4) {$0_{c}$};
			\node[d dot] (dl) at (-1,-5) {$0_d$};
			\node[d dot] (dr) at (1,-5) {$0_d$};
			\node[e dot] (el) at (-1,-6) {$0_e$};
			\node[e dot] (er) at (1,-6) {$1_e$};
			
			\draw[-] (a) -- node[left] {$X'$} (c);
			\draw[-] (c) -- node[above,pos=.7] {$Y'$} (bl);
			\draw[-] (c) -- (br);
			\draw[-] (bl) -- (bbl);
			\draw[-] (br) -- (bbr);
			\draw[dotted,thick] (bbl) -- (bbbl);
			\draw[dotted,thick] (bbr) -- (bbbr);
			\draw[-] (bbbl) -- (dl);
			\draw[-] (bbbr) -- (dr);
			\draw[-] (dl) -- (el);
			\draw[-] (dr) -- (er);
		\end{tikzpicture} \qquad
		\begin{tikzpicture}
			\node (Complex) at (-1.5,-.5) {$\C'':$};
			\node[a dot] (a) at (0,0.3) {$0_a$};
			\node[c dot] (c) at (0,-1) {$1_{b}$};
			\node[b dot] (bl) at (-1,-2) {$0_{c}$};
			\node[b dot] (br) at (1,-2) {$0_{c}$};
			\node[d dot] (bbl) at (-1,-3) {$0_{d}$};
			\node[d dot] (bbr) at (1,-3) {$0_{d}$};
			\node[b dot] (bbbl) at (-1,-4) {$0_{c}$};
			\node[b dot] (bbbr) at (1,-4) {$0_{c}$};
			\node[d dot] (dl) at (-1,-5) {$0_d$};
			\node[d dot] (dr) at (1,-5) {$0_d$};
			\node[e dot] (er) at (1,-6) {$1_e$};
			
			\draw[-] (a) -- node[left] {$X''$} (c);
			\draw[-] (c) -- node[above,pos=.7] {$Y''$} (bl);
			\draw[-] (c) -- (br);
			\draw[-] (bl) -- (bbl);
			\draw[-] (br) -- (bbr);
			\draw[dotted,thick] (bbl) -- (bbbl);
			\draw[dotted,thick] (bbr) -- (bbbr);
			\draw[-] (bbbl) -- (dl);
			\draw[-] (bbbr) -- (dr);
			\draw[-] (dr) -- (er);
		\end{tikzpicture}
}
	\caption{\rev{Differences between~$(\C,X)$, $(\C',X')$,~and $(\C'',X'')$ are not local.}}
\label{fig.localglobal}
\end{figure}
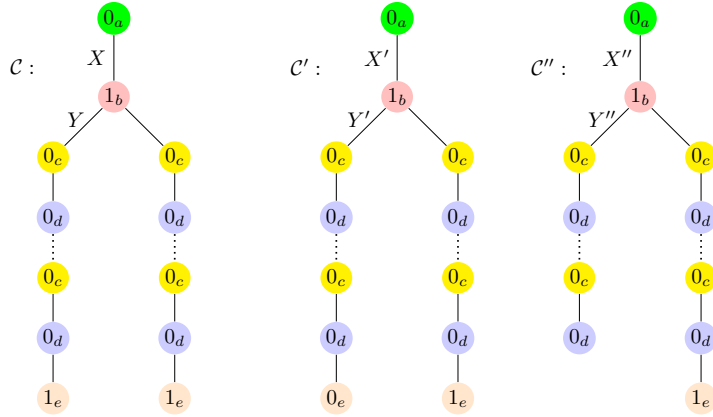

\begin{example}\label{ex.faraway}
\rev{Consider simplicial models~$\C$, $\C'$,~and~$\C''$ in Fig.~\ref{fig.localglobal}, where the number of yellow $0_c$~vertices (purple $0_d$~vertices) on each vertical branch is~$m$. 
The three models differ only  in the \rev{$e$}-labeled vertices at the bottom, which can be reached from facets~$X$, $X'$,~and~$X''$ respectively in no fewer than $2m+1$~consecutive steps. 
Despite all facets within at most $2m$~steps from~$X$, $X'$,~and~$X''$ coinciding, we have $(\C,X)\not\equiv_{\localK}(\C',X')$ because~$\C,X \fctsat K_b (K_cK_d)^m  p_e$ whereas~$\C',X' \nfctsat K_b (K_cK_d)^m p_e$. 

Assume towards a contradiction that there exists some notion~$\bisim^\ast$ of bisimulation that allows to ignore some of the adjacent facets in \textbf{forth}/\textbf{back} steps based on some local condition and that this notion satisfies the Hennessy--Milner property. 
In~particular, $(\C,X)\nbisim^\ast(\C',X')$. 
Since the only difference between~$\C$~and~$\C'$ can be reached through~$Y$~and~$Y'$ respectively, $\bisim^\ast$~cannot avoid the \textbf{back} step from~$X'$ to~$Y'$. 
At~the same time, $Y''$~in~$\C''$ is the same as~$Y'$ in~$\C'$, and so are all facets reachable by at most $2m-1$~consecutive steps from~$Y''$ respectively~$Y'$. 
Given that $m$~can be made arbitrarily large, no local condition can differentiate~$Y''$ from~$Y'$, meaning that $\bisim^\ast$,~in testing  $(\C,X)$~and $(\C'',X'')$,  must also consider~$Y''$ as a \textbf{back} step from~$X''$. 
It~is easy to see that $Y''$~is not modally equivalent to any facet of~$\C$. 
Indeed, $\C'', Y'' \fctsat p_b \land \neg p_c$, yet $\C'', Y'' \nfctsat (K_{c} K_{d})^m p_e$. 
Thus, $(\C,X)\nbisim^\ast(\C'',X'')$ must hold due to the failure of \textbf{back} to~$Y''$. 
At the same time, one can show, as in Prop.~\ref{prop.llll}, that $(\C,X)\equiv_{\localK}(\C'',X'')$. 
This~contradiction with the  Hennessy--Milner property shows that no reasonable standard bisimulation~$\bisim^\ast$ exists  for language~$\localK$.}
\end{example}

\rev{What is the reason for the failure of Hennessy--Milner property for language~$\localK$? 
Kupke and Pattinson,  in the context of coalgebraic semantics of modal logics~\cite{KupkeP11}, call the Hennessy--Milner property the~``gold~standard'' regarding the question of~``expressivity of modal logic,'' or, more specifically, ``whether the logic~[...]~provide[s] enough power to describe particular properties of a system.'' 
Similarly, and specifically in the context of many-valued bisimulations, Marti and Metcalfe~\cite{Marti2014-MARAHP-6} equate the formal question of ``whether analogues of the Hennessy–Milner property~[...]~hold for image-finite models of many-valued modal logics'' with the informal attempt ``to~determine whether the language is expressive enough to distinguish image-finite models of many-valued modal logics.'' 
Combining these two separate areas, in studying many-valued modal logics coalgebraically, B\'{\i}lkov\'a and Dost\'al~\cite{BilkovaDostal2016} ask whether ``finitary modal languages [are] expressive for bisimilarity'' by checking the Hennessy–Milner property. 
Thus, the proper question is not why Hennessy–Milner fails, but what properties language~$\localK$ is not able to express. 
Put~this way, the answer is not hard to guess. 
What~is missing is the ability to talk about agents being dead/alive in the object language. 
In~the next section, we show that  adding global atoms for agents being alive indeed remedies the situation.} 

\section{Bisimulation for Simplicial Models with Glocal Atoms}
 \label{sec.bisimulation}
In~this section, we switch to language~$\glocalK$ with global atoms. 
Hence,~$\equiv$~from now on is understood as~$\equiv_{\glocalK}$.
Note that the \textbf{atoms} clause of Def.~\ref{def:bis0} requires that not only local \rev{atoms} have the same truth values, i.e.,~$\ell(X) =\ell'(X')$, but that global \rev{atoms} do too, i.e.,~\mbox{$\chi(X)= \chi'(X')$}. 

\begin{example} \label{example.bisim}
The~three simplicial models below are pairwise  
bisimilar. 
Moreover, there exist bisimulations that relate each facet of one model to some facet of another in either direction (we~call such bisimulations \emph{total}). 
In~particular,
$\{(X, X'), (X,Z'), (Y,Y')\}$~is a bisimulation between the left and middle models. 
Similarly, $\{(X, X''), (Y,Y''), (Y,Z'')\}$~is a bisimulation between the left and right models.
The~atomic harmony for global atoms is a consequence of only relating facets with 
 the same set of agents.
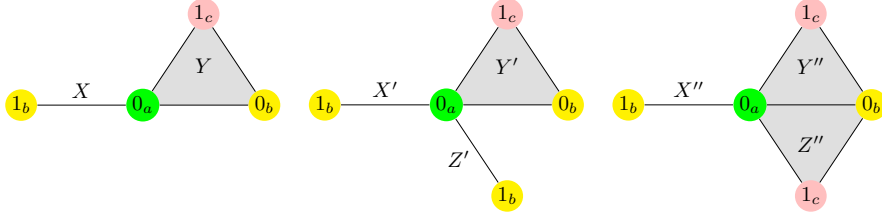
\begin{center}
\scalebox{.8}{
\begin{tikzpicture}
\fill[fill=gray!25!white] (2,0) -- (4,0) -- (3,1.51) -- cycle;
\node[b dot] (b1) at (0,0) {$1_b$};
\node[b dot] (b0) at (4,0) {$0_b$};
\node[c dot] (c1) at (3,1.51) {$1_c$};
\node[a dot] (a0) at (2,0) {$0_a$};
\node (f1) at (3,.65) {$Y$};
\draw[-] (b1) -- node[above] {$X$} (a0);
\draw[-] (a0) -- (b0);
\draw[-] (b0) -- (c1);
\draw[-] (a0) -- (c1);
\node (c1b) at (3,-1.51) {\color{white}$1_b$};
\end{tikzpicture}
\quad
\begin{tikzpicture}
\fill[fill=gray!25!white] (2,0) -- (4,0) -- (3,1.51) -- cycle;
\node[b dot] (b1) at (0,0) {$1_b$};
\node[b dot] (b0) at (4,0) {$0_b$};
\node[c dot] (c1) at (3,1.51) {$1_c$};
\node[a dot] (a0) at (2,0) {$0_a$};
\node (f1) at (3,.65) {$Y'$};
\draw[-] (b1) -- node[above] {$X'$} (a0);
\draw[-] (a0) -- (b0);
\draw[-] (b0) -- (c1);
\draw[-] (a0) -- (c1);
\node[b dot] (c1b) at (3,-1.51) {$1_b$};
\draw[-] (a0) --(c1b)  node at (2.2,-.9) {$Z'$};
\end{tikzpicture}
\quad
\begin{tikzpicture}[round/.style={circle,fill=white,inner sep=1}]
\fill[fill=gray!25!white] (2,0) -- (4,0) -- (3,1.51) -- cycle;
\fill[fill=gray!25!white] (2,0) -- (4,0) -- (3,-1.51) -- cycle;
\node[b dot] (b1) at (0,0) {$1_b$};
\node[b dot] (b0) at (4,0) {$0_b$};
\node[c dot] (c1) at (3,1.51) {$1_c$};
\node[a dot] (a0) at (2,0) {$0_a$};
\node (f1) at (3,.65) {$Y''$};
\draw[-] (b1) -- node[above] {$X''$} (a0);
\draw[-] (a0) -- (b0);
\draw[-] (b0) -- (c1);
\draw[-] (a0) -- (c1);
\node[c dot] (c1b) at (3,-1.51) {$1_c$};
\node (f1) at (3,-.65) {$Z''$};
\draw[-] (b0) -- (c1b);
\draw[-] (a0) -- (c1b);
\end{tikzpicture}
}
\end{center}
\end{example}

The~following theorem is the easier direction of the Hennessy--Milner property: bisimilarity implies modal equivalence.
\begin{theorem}[Bisimilarity implies modal equivalence]
	\label{BisimImpliesEquiv}
	For~arbitrary pointed simplicial models\/ $(\C,X)$~and\/ $(\C',X')$, 
\[
	(\C,X)\bisim(\C',X')
	\quad\Longrightarrow\quad
	(\C,X)\equiv_{\glocalK}(\C',X').
	\]
\end{theorem}
\begin{proof}
	Let~$(\C,X)\bisim(\C',X')$ for~$\C=(C,\chi,\ell)$~and $\C'=(C',\chi', \ell')$.
	Let~$X \B X'$ for a bisimulation~$\B$ between~$\C$~and~$\C'$. 
	By~Lemma~\ref{simplifyModalEquivalence}, to prove~\mbox{$(\C,X) \equiv (\C',X')$},  it is sufficient to prove~\eqref{eq:modeq_true} for~$Y=X$~and $Y'=X'$ for all~$\varphi \in \glocalK$. 
	We prove both~\eqref{eq:modeq_true}~and~\eqref{eq:modeq_def} for all facets~$Y \in \FF(C)$~and $Y' \in \FF(C')$ such that~$Y \B Y'$ by mutual induction  on the construction of~$\varphi\in \glocalK$. 
	\begin{compactitem}
		\item 
			For~global atom~$\varphi = a$:
			\eqref{eq:modeq_def}~is trivial 
			because $a$~is defined  on all facets.  
			\eqref{eq:modeq_true}~follows from~$\chi(Y)=\chi'(Y')$ part of  \textbf{atoms}.
			
		\item 
			For~local atom~$\varphi = p_a$:  \eqref{eq:modeq_def}~and~\eqref{eq:modeq_true}~follow from~$\chi(Y)=\chi'(Y')$~and $\ell(Y)=\ell'(Y')$ parts of  \textbf{atoms} respectively. 	
			
		\item 
			For~negation~$\varphi = \neg \psi$: 
			\begin{compactitem}
				\item[\eqref{eq:modeq_def}] 
				$
				\C,Y \fctdef \neg \psi \Longleftrightarrow
				\C,Y \fctdef \psi \xLeftrightarrow{\text{IH\eqref{eq:modeq_def}}}
				\C',Y' \fctdef  \psi \Longleftrightarrow
				\C',Y' \fctdef \neg \psi 
				$.
				
				\item[\eqref{eq:modeq_true}] 
				\label{see_negation}
				$
				\C,Y \fctsat \neg \psi 
				\Longleftrightarrow
				(\C, Y \fctdef \psi \text{ and }\C,Y \nfctsat \psi) \xLeftrightarrow{\text{IH\eqref{eq:modeq_def},\eqref{eq:modeq_true}}}$\\
				$\phantom{\C,Y \fctsat \neg \psi 
					\Longleftrightarrow{}} (\C', Y' \fctdef \psi \text{ and }\C',Y' \nfctsat \psi) 
				\Longleftrightarrow
				\C',Y' \fctsat \neg \psi 
				$.
			\end{compactitem}
			
		\item 
			For~conjunction~$\varphi = \psi \land \chi$, the argument is similar.
			
		\item 
			For~modality~$\varphi = \M_a \psi$, the argument for~\eqref{eq:modeq_def}~and~\eqref{eq:modeq_true} is similar, thus, we only prove~\eqref{eq:modeq_def}.
			From~left to right, 
			assume~$\C,Y \fctdef \M_a \psi$. 
			Then~there is~$Z \in \FF(C)$ such that~$a \in \chi(Y \cap Z)$~and $\C, Z \fctdef \psi$.  
			By~\textbf{forth} for~$Y \B Y'$, there is~$Z' \in \FF(C')$ such that~$a \in \chi'(Y' \cap Z')$~and $Z \B Z'$. 
			Then~$\C',Z' \fctdef \psi$ by the~IH\eqref{eq:modeq_def} for~$\psi$~and $Z \B Z'$. 
			Thus,~$\C',Y' \fctdef \M_a \psi$. 
			The~other direction is symmetric using \textbf{back} instead of \textbf{forth}.
		\end{compactitem}
			\vspace{-.45cm}
\end{proof}

Before proving the other, harder direction of the Hennessy--Milner property, we first explain where the standard proof fails. 
In~order to show the converse, i.e.,~that modal equivalence implies bisimilarity, one typically defines a binary relation of modal equivalence between states of two given models and shows that it is a bisimulation, \rev{under a suitable finitary assumption. 
As~mentioned earlier, the standard assumption for Kripke models is that of image-finiteness~\cite{Marti2014-MARAHP-6}. 
The~translation of this assumption into simplicial models yields the property of  \emph{star-finiteness}}: a simplicial model is \emph{star-finite} if{f}, for each agent~$a$, the number of facets $a$-adjacent to any given facet is finite.
Assuming star-finite $(\C,X)$~and $(\C',X')$ such that~$(\C,X) \equiv (\C',X')$, consider \textbf{forth}.  
Let~$Y$~be $a$-adjacent to~$X$ in~$\C$, and assume, towards a contradiction, that there is no modally equivalent~$Y'$ in~$\C'$ that is $a$-adjacent to~$X'$. 
Then,~for each facet $a$-adjacent to~$X'$ in~$\C'$, there must be  a formula distinguishing it from~$Y$. 
By~the finitary assumption, there are only finitely many such facets~$Y'_1,\dots,Y'_k$ and formulas~$\psi_1,\dots, \psi_k$ such that $\psi_i$~distinguishes~$Y$ from~$Y'_i$. 
Since a formula distinguishes facets if{f} its negation does, in \rev{the} two-valued semantics, w.l.o.g.,~we may assume that each~$\psi_i$ is true in~$Y$ and false in~$Y'_i$, so that~$\M_a(\psi_1\land\dots\land \psi_k)$ is true in~$X$ and false in~$X'$, contradicting the assumption of their modal equivalence. 
By~contrast,  in \rev{the} three-valued semantics 
some of~$\psi_i$'s may be undefined in~$Y$, which cannot be sidestepped by switching to the negation.

Thus, to implement the same method, instead of negation, one needs a more complex way of transforming a distinguishing formula~$\psi_i$ that is undefined in~$Y$ but defined in~$Y'_i$ into another distinguishing formula~$f(\psi_i)$ that is true in~$Y$ and  false or undefined in~$Y'_i$. 
Unfortunately, such transformation cannot be done independently of~$(\C,X)$, as follows from the following proposition:

\begin{proposition}
\label{rest:nodirecttrans}
For~formula~$\varphi = p_a \et K_b p_c$, no formula~$\varphi^{\nfctdef} \in \glocalK$ satisfies 
\begin{gather}
	\C, X \nfctdef \varphi 
 	\qquad\Longleftrightarrow \qquad
 	\C,X \fctsat \varphi^{\nfctdef}
 \label{eq:simple_trans_ex}
\end{gather}
for all pointed simplicial models\/~$(\C,X)$.
\end{proposition}
\begin{proof}
	\rev{For~any $\varphi^{\nfctdef} \in \glocalK$, there exists a unique purely propositional formula~$\xi(x_1,\dots,x_m)$, with~$m\geq 0$, such that $\varphi^{\nfctdef} = \xi\bigl(\M_{i_1}\psi_1,\dots,\M_{i_m} \psi_m\bigr)$ for some agents~$i_1,\dots, i_m \in A$ and formulas~$\psi_1,\dots, \psi_m \in \glocalK$. 
	We~consider three cases: 
	(i)~$i_j=a$ for some~$1 \leq j \leq m$,  
	(ii)~$i_j\ne a$  for some~$1 \leq j \leq m$,~and 
	(iii)~$m=0$, i.e.,~$\varphi^{\nfctdef}=\xi$ is purely propositional.} 
	For~each case, we construct $\C= (C, \chi, \ell)$~and $X \in \FF(C)$ to violate~\eqref{eq:simple_trans_ex}.
	\begin{compactenum}[(i)]
		\item
			$\C, X \nfctdef p_a \et K_b p_c$ if~$a \notin \chi(X)$, but~$\C, X \nfctdef \varphi^{\nfctdef}$ because of~$\M_a \rev{\psi_j}$. 
		\item
			$\C, X \nfctdef p_a \et K_b p_c$ if~$ \chi(X)=\{a\}$, but~$\C, X \nfctdef \varphi^{\nfctdef}$ because of~$\rev{\M_{i_j} \psi_j}$. 
		\item
			Consider~$(\C,X)$~and $(\C',X')$ depicted below with the same evaluations of all~$p_a \in P_a$ on both $a$-colored vertices and of all~$p_b \in P_b$ on both $b$-colored ones.
			Then~$\C, X \nfctdef p_a \et K_b p_c$ while~$\C', X' \fctdef p_a \et K_b p_c$. 
			Thus, according to~\eqref{eq:simple_trans_ex}, any~$\varphi^{\nfctdef}$ should distinguish $(\C,X)$ from~$(\C',X')$. 
			But~it is easy to prove by construction \rev{that any purely propositional  formula~$\xi$} has the same truth value (true,~false,~or undefined) in~$(\C,X)$ as in~$(\C',X')$. 
		\begin{center}
			\scalebox{.8}{
				\begin{tikzpicture}
					\node (c) at (-1,0) {$\C:$};
					\node[b dot] (b1) at (0,0) {$a$};
					\node[a dot] (a0) at (2,0) {$b$};
					\draw[-] (b1) -- node[above] {$X$} (a0);
				\end{tikzpicture}
				\qquad\qquad
				\begin{tikzpicture}
					\node (c) at (-1,0) {$\C':$};
					\node[b dot] (b1) at (0,0) {$a$};
					\node[c dot] (b0) at (4,0) {$c$};
					\node[a dot] (a0) at (2,0) {$b$};
					\draw[-] (b1) -- node[above] {$X'$} (a0);
					\draw[-] (a0) -- (b0);
				\end{tikzpicture}
			}
		\end{center}
	\end{compactenum}%
	\vspace{-.48cm}
\end{proof}

Since a purely syntactical transformation from undefined to true is impossible, we will present such a transformation modulo  a given~$(\C,X)$
(Lemma~\ref{formula.for.Non.embedded.LT}). 
The~construction relies on the modal structure of formula~$\varphi$, which is encoded in a form of what we call  a \emph{life tree}. 
With~this localized transformation in place, we will be able to complete the proof of the remaining direction of Hennessy--Milner along the lines of the standard proof (Prop.~\ref{hehelastone}).

\begin{definition}
A~\emph{life tree}~$\lt = (\Tree,\treeLab)$ is a directed rooted tree~$\Tree=(V, E)$  
 supplied with a function~$\treeLab$ that labels each node~$\sigma \in V$ with a set~$\treeLab(\sigma) \subseteq A$ of agents and each edge~$(\sigma,\tau) \in E$ with an agent~$\treeLab(\sigma,\tau)\in A$ such that $\treeLab$~satisfies the  property
 $
 \treeLab(\sigma, \tau) \in \treeLab(\sigma) \cap \treeLab(\tau)
 $
 for each~$(\sigma,\tau) \in E$.
 We~call~$(\sigma,\tau)$ an  \emph{$a$-edge} if{f} $\treeLab(\sigma,\tau) = a$. 
 For~the root of a life tree, we typically use~$\rho$, possibly with a subscript, for instance, $\rho_1$~could be used for the root of~$\Tree_1$.   
 The~\emph{$a$-grafting~$\lt^a$} of a life tree~$\lt$ is obtained by adding~$a$ to the label of its root.
 \end{definition}

\begin{definition}
\label{def:lifeform}
The \emph{life tree~$\ltf{\varphi}$ of a formula~$\varphi \in \glocalK$} is a labeled tree~$(\Tree,\treeLab)$ defined by recursion on the construction of~$\varphi$:
	\begin{compactitem}
		\item $\ltf{a}$~is the  tree consisting of a single root~$\rho$ with~$\treeLab(\rho)=\varnothing$;
		\item $\ltf{p_a}$~is the  tree consisting of a single root~$\rho$ with~$\treeLab(\rho)=\{a\}$;
		\item $\ltf{\neg \varphi} \colonequals \ltf{\varphi}$;
		\item $\ltf{\varphi \land \psi}$~is obtained by merging the roots~$\rho_\varphi$ of~$\ltf{\varphi}=(\Tree_\varphi,\treeLab_\varphi)$~and $\rho_\psi$~of~$\ltf{\psi}=(\Tree_\psi,\treeLab_\psi)$ and their labels, i.e.,~by first taking the disjoint union of~$\ltf{\varphi}$~and~$\ltf{\psi}$ with all the labeling preserved and then merging both roots~$\rho_\varphi$~and~$\rho_\psi$ into a new root~$\rho$ with the label~$\treeLab_\varphi(\rho_\varphi) \cup \treeLab_\psi(\rho_\psi)$ without changing labels of any edges or of any other nodes;
		\item $\ltf{\M_a \varphi}$~is obtained by adding a  new root~$\rho$ to  the $a$-grafting~$\ltf{\varphi}^a$ of~$\ltf{\varphi}$, labeling this new root~$\rho$ with~$\{a\}$,  and adding a new $a$-edge from~$\rho$ to the root~$\rho_\varphi$ of~$\ltf{\varphi}^a$, with~$\treeLab$~taking over all labels of~$\ltf{\varphi}^a$ without change.
	\end{compactitem}
\end{definition}

It~is easy to show by induction on the construction of formula life trees that

\begin{proposition}[Correctness for formula life trees]
\label{prop:life_tree}
The~life tree\/~$\ltf{\varphi}$ 
of any formula~$\varphi \in \glocalK$ is a life tree. 
In~addition,
for any child~$\sigma$ of the root  of~$\Tree$,  the subtree of\/~$\ltf{\varphi}$ rooted in~$\sigma$ is the $a$-grafting\/~$\ltf{\psi}^a$ of the life tree\/~$\ltf{\psi}$ of  some subformula~$\psi$ of~$\varphi$ for some agent~$a$.
\end{proposition}

\begin{example}
	\label{eg:formula_and_life_tree}
The~life trees of three formulas are depicted in Fig.~\ref{fig:sample_life_trees}.  
Note~that $\ltf{\M_b \neg p_d \land \M_c p_d}^a$~is a subtree of~$\ltf{\M_a(\M_b \neg p_d \land \M_c p_d)}$. 

	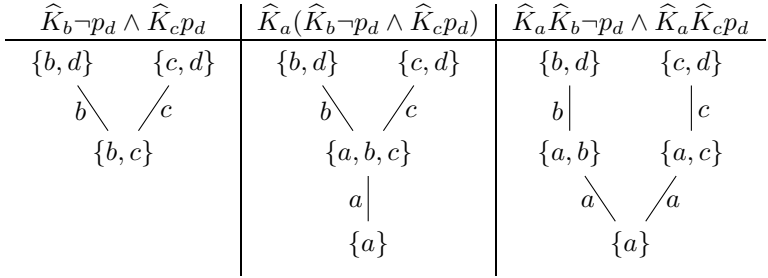
\begin{figure}[h]
	\centering
		\begin{tabular}{c|c|c}
			$\M_b \neg p_d \land \M_c p_d$ & 
			$\M_a(\M_b \neg p_d \land \M_c p_d)$ & 
			$\M_a\M_b \neg p_d \land \M_a \M_c p_d$ 
			\\
			\hline
			\begin{tikzpicture}[scale=0.8]
				\node (root) at (0,0) {\phantom{$\{a\}$}};
				\node (L1) at (0,1.5) {$\{b,c\}$};
				\node (L2l) at (-1,3) {$\{b,d\}$};
				\node (L2r) at (1,3) {$\{c,d\}$};
				
				\draw (L1)--(L2l)  node at (-0.7,2.2) {$b$};
				\draw (L1)--(L2r)  node at (0.7,2.2) {$c$};
			\end{tikzpicture}
			&
			\begin{tikzpicture}[scale=0.8]
				\node (root) at (0,0) {$\{a\}$};
				\node (L1) at (0,1.5) {$\{a,b,c\}$};
				\node (L2l) at (-1,3) {$\{b,d\}$};
				\node (L2r) at (1,3) {$\{c,d\}$};
				
				\draw (root)--(L1)  node at (-0.2,0.7) {$a$};
				\draw (L1)--(L2l)  node at (-0.7,2.2) {$b$};
				\draw (L1)--(L2r)  node at (0.7,2.2) {$c$};
			\end{tikzpicture}
			&
			\begin{tikzpicture}[scale=0.8]
				\node (root) at (0,0) {$\{a\}$};
				\node (L1l) at (-1,1.5) {$\{a,b\}$};
				\node (L1r) at (1,1.5) {$\{a,c\}$};
				\node (L2l) at (-1,3) {$\{b,d\}$};
				\node (L2r) at (1,3) {$\{c,d\}$};
				
				\draw (root)--(L1l) node at (-0.7,0.7) {$a$};
				\draw (root)--(L1r) node at (0.7,0.7) {$a$};
				\draw (L1l)--(L2l) node at (-1.2,2.2) {$b$};
				\draw (L1r)--(L2r) node at  (1.2,2.2) {$c$};
			\end{tikzpicture}
		\end{tabular}
		\caption{Life trees of three sample formulas}
		\label{fig:sample_life_trees}
	\end{figure}
\end{example}

\begin{definition}
	\label{def:embedding}
	An~\emph{embedding} of a  life tree~$\lt= (\Tree,\treeLab)$  with~$\Tree=(V,E)$  into~$(\C,X)$ with~$\C = (C, \chi, \ell)$ is a function~$e \colon V \to \FF(C)$  such that 
\begin{compactenum}
\item\label{emb:root} 
	the~root~$\rho$ of~$\Tree$ is mapped to~$X$, i.e.,~$e(\rho)= X$;
\item\label{emb:node} 
	$\treeLab(\sigma)\subseteq \chi\bigl(e(\sigma)\bigr)$ for all tree nodes~$\sigma\in V$;
\item\label{emb:edge} 
	$\treeLab(\sigma,\tau) \in \chi\bigl(e(\sigma) \cap e(\tau)\bigr)$ for all tree edges~$(\sigma,\tau)\in E$.
\end{compactenum}
Then~we write~$\C,X \ltdef_e \lt$, or simply~$\C,X \ltdef \lt$\rev{,} if such an embedding exists.
\end{definition}

The~following two lemmas show that formula life trees provide a practical and compact syntactical description of definability.

\begin{lemma}[Equivalence of embeddability and definability]
\label{EmbeddingAsDefinability}
For~any formula~$\varphi \in \glocalK$ and pointed simplicial model\/~$(\C,X)$,
\begin{equation}
\label{eq:embed_life}
\C, X \fctdef \varphi 
\quad\Longleftrightarrow\quad 
\C,X \ltdef \ltf{\varphi}.
\end{equation}
\end{lemma}
\begin{proof}
	Let~$\C = (C, \chi, \ell)$ and $\ltf{\varphi} = \bigl((V,E),\treeLab\bigr)$ with root~$\rho \in V$.
	We~prove~\eqref{eq:embed_life} for all $X \in \FF(C)$ by induction on \rev{a formula}~$\varphi \in \glocalK$.  
	\begin{compactitem}
		\item For~$\varphi=a$: 
		Since $\C, X \fctdef a$ always holds, we need to show that~$\C, X \ltdef \ltf{a}$. 
		By~Def.~\ref{def:embedding}\eqref{emb:root}, no function other than~$e= \{(\rho,X)\}$ can be an embedding  of~$\ltf{a}$ into~$(\C,X)$. 
		To show that it is, given that $\ltf{a}$~has no edges, it suffices to check Def.~\ref{def:embedding}\eqref{emb:node}, which holds since~$\treeLab(\rho) = \varnothing \subseteq \chi(X)$.
		
		\item For~$\varphi=p_a$: 
		Again, Def.~\ref{def:embedding}\eqref{emb:root} means that~$e(\rho)=X$ and there is at most one embedding of~$\ltf{p_a}$ into~$(\C,X)$. 
		Since in this case $\treeLab(\rho) = \{a\}$, it remains to note that~$\C, X \fctdef p_a$ if{f} $a \in \chi (X)$ if{f} $\treeLab (\rho) \subseteq \chi(X)$ if{f} $\C, X \ltdef \ltf{p_a}$.
		
		\item For~$\varphi=\neg \psi$:  
		$\C,X \fctdef \neg \psi$ if{f} $\C,X \fctdef \psi$ if{f} by the~IH $\C,X \ltdef \ltf{\psi}$~if{f} $\C,X \ltdef \ltf{\neg \psi}$.
		
		\item For~$\varphi=\psi \land \xi$: 
		By~Def.~\ref{lanKa.fct.defsat},  $\C,X \fctdef \psi \land \xi$ if{f} $\C,X \fctdef \psi$~and $\C,X \fctdef \xi$. 
		By~the~IH, this is equivalent to~$\C,X \ltdef \ltf{\psi}$~and $\C,X \ltdef \ltf{\xi}$. 
		It~remains to  prove that  this is equivalent to~$\C,X \ltdef \ltf{\psi \land \xi}$. 
		Let~\rev{the} life trees~$\ltf{\psi}=\bigl((V_\psi,E_\psi),\treeLab_\psi\bigr)$~and $\ltf{\xi}=\bigl((V_\xi,E_\xi),\treeLab_\xi\bigr)$ have roots~$\rho_\psi$~and~$\rho_\xi$ respectively.
		
		Assume first that~$\C,X \ltdef_{e_\psi} \ltf{\psi}$~and $\C,X \ltdef_{e_\xi} \ltf{\xi}$ for~$e_\psi\colon V_\psi \to \FF(C)$~and $e_\xi\colon V_\xi \to \FF(C)$. 
		Note~that~$e_\psi(\rho_\psi) = e_\xi(\rho_\xi) = X$ by Def.~\ref{def:embedding}\eqref{emb:root}. 
		It~is easy to see that~$\C,X \ltdef_e \ltf{\psi \land \xi}$ for \rev{the} function~$e \colon V \to \FF(C)$ defined~by
		\[
		e(\theta) 
		\ce
		\begin{cases}
			X & \text{ if $\theta = \rho$},
			\\
			e_\psi(\theta) & \text{ if  $\theta \in V_\psi \setminus \{\rho_\psi\}$},
			\\
			e_\xi(\theta) & \text{ if  $\theta \in V_\xi \setminus \{\rho_\xi\}$}.
		\end{cases}
		\] 
		
		Assume now, conversely, that~$\C,X \ltdef_e \ltf{\psi \land \xi}$ for~$e \colon V \to \FF(C)$. 
		It~is easy to see that~$\C,X \ltdef_{e_\psi} \ltf{\psi}$~and $\C,X \ltdef_{e_\xi} \ltf{\xi}$ for~$e_\psi \colon V_\psi \to \FF(C)$~and $e_\xi \colon V_\xi \to \FF(C)$ defined~by	
		\[
		e_\psi(\sigma) 
		\ce
		\begin{cases}
			X & \text{ if $ \sigma=r_\psi$},
			\\
			e(\sigma) & \text{ if $\sigma \in V_\psi \setminus \{r_\psi\}$},
		\end{cases}
		\qquad
		e_\xi(\tau) 
		\ce
		\begin{cases}
			X & \text{ if $ \tau= \rho_\xi$},
			\\
			e(\tau) & \text{ if $\tau \in V_\xi \setminus \{r_\xi\}$}.
		\end{cases}
		\] 
		
		\item For~$\varphi = \M_a \psi$: 
		By Def.~\ref{lanKa.fct.defsat},  $\C,X \fctdef \M_a \psi$ if{f} $\C,Y \fctdef \psi$ for some~$Y \in \FF(C)$ with~$a \in \chi(X \cap Y)$. 
		By~the~IH,  the latter is equivalent to~$\C,Y \ltdef \ltf{\psi}$ for some~$Y \in \FF(C)$ with~$a \in \chi(X \cap Y)$. 
		It remains to prove that this is equivalent to~$\C,X \ltdef \ltf{\M_a \psi}$. 
		Let \rev{the} life tree $\ltf{\psi}=\bigl((V_\psi,E_\psi),\treeLab_\psi\bigr)$  have root~$\rho_\psi$, which is the only child of~$\rho$ in~$\ltf{\M_a \psi}$, where~$V = V_\psi \sqcup \{\rho\}$.
		
		Assume first that~$\C,X \ltdef_e \ltf{\M_a \psi}$, for some~$e\colon V \to \FF(C)$. 
		By~Def.~\ref{def:embedding}\eqref{emb:root},  $e(\rho)=X$.  
		For~$Y = e(\rho_\psi)$, it is easy to see that~$\C, Y \ltdef_{e \upharpoonright V_\psi} \ltf{\psi}^a$
		and, hence, $\C, Y \ltdef_{e \upharpoonright V_\psi} \ltf{\psi}$.  
		Finally, we have $a = \treeLab(\rho, \rho_\psi) \in \chi\bigl(e(\rho) \cap e(\rho_\psi)\bigr)$ by Def.~\ref{def:embedding}\eqref{emb:edge}, i.e.,~$a \in \chi(X \cap Y)$. 
		
		Assume now, conversely, that~$\C,Y \ltdef_{e_\psi} \ltf{\psi}$ for some~$Y \in \FF(C)$ with~\mbox{$a \in \chi(X \cap Y)$}~and $e_\psi \colon V_\psi \to \FF(C)$. 
		Let~us show that~$\C,X \ltdef_{e} \ltf{\M_a \psi}$ for~$e \ce e_\psi \sqcup \{(\rho,X)\}$. 
		Indeed, $e(\rho) = X$ fulfilling Def.~\ref{def:embedding}\eqref{emb:root}. 
		To~show  Def.~\ref{def:embedding}\eqref{emb:node} for~$\rho$, note that~$\treeLab(\rho) = \{a\} \subseteq \chi(X\cap Y) \subseteq \chi(X)= \chi\bigl(e(\rho)\bigr)$. 
		$\treeLab(\rho, \rho_\psi) = a \in \chi(X \cap Y) =  \chi\bigl(e(\rho) \cap e(\rho_\psi)\bigr)$ proves Def.~\ref{def:embedding}\eqref{emb:edge} for~$(\rho,\rho_\psi)$. 
		\[
		\treeLab(\rho_\psi) = \{a\}\cup\treeLab_\psi(\rho_\psi) \subseteq \chi(X\cap Y)  \cup \chi\bigl(e_\psi(\rho_\psi)\bigr)\subseteq \chi(Y) \cup \chi(Y)= \chi\bigl(e(\rho_\psi)\bigr)
		\] 
		shows Def.~\ref{def:embedding}\eqref{emb:node} for~$\rho_\psi$.
		Def.~\ref{def:embedding}\eqref{emb:node}--\eqref{emb:edge}   for the remaining  nodes and edges easily follow from the same properties of~$e_\psi$.
		Thus, $\C,X \ltdef  \ltf{\M_a \psi}$.
	\end{compactitem}
	\vspace{-.45cm}
\end{proof}

\begin{lemma}[Criterion of non-embeddability]
	\label{Non.embedded.Tree}
	Let\/~$(\C,X)$ be a pointed simplicial model with~$\C = (C, \chi, \ell)$ and\/ $\ltf{\varphi} = \bigl((V,E),\treeLab\bigr)$ be the life tree of a formula~$\varphi \in \glocalK$ with root~$\rho$. 
	Then~$\C,X\nltdef\ltf{\varphi}$ \rev{if{f}} at least one of the following two criteria is satisfied:
	\begin{compactenum}
	\item\label{critone} 
		either\/~$\treeLab(\rho) \nsubseteq \chi(X)$~or 
	\item\label{crittwo} 
		there is a subformula~$\psi$ of~$\varphi$ and an $a$-edge\/~$(\rho,\rho_\psi)\in E$ to the root~$\rho_\psi$  of $a$-grafting\/~$\ltf{\psi}^a$ such that~$\C,Y\nltdef\ltf{\psi}$ for any~$Y \in \FF(C)$ with~$a \in \chi(X \cap Y)$.
	\end{compactenum}
\end{lemma}
\begin{proof}
Let~$(\rho,\sigma_1),\dots,(\rho,\sigma_k)\in E$ be all edges from~$\rho$  in~$\ltf{\varphi}$, where~$k \geq 0$ and let~$\treeLab(\rho,\sigma_i)=a_i$.
By Prop.~\ref{prop:life_tree}, each~$\sigma_i$ is the root of~$\ltf{\psi_i}^{a_i}=\bigl((V_i,E_i),\treeLab_i\bigr)$ for some subformula~$\psi_i$ of~$\varphi$. 
In particular, $V = \{\rho\} \sqcup \bigsqcup_{i=1}^k V_i$.

First,~let  $\C,X\nltdef\ltf{\varphi}$, and 
	assume towards a contradiction that both criteria~\eqref{critone}~and~\eqref{crittwo}  fail. 
	Then~for each~$\psi_i$, there is~$Y_i \in \FF(C)$ with~$a_i \in \chi(X \cap Y_i)$ such that~$\C,Y_i	\ltdef_{e_i} \ltf{\psi_i}$ for 
	some~$e_i \colon V_i \to \FF(C)$ with~$e_i(\sigma_i) = Y_i$.
	It~is easy to show that~$\C,X \ltdef_e \ltf{\varphi}$ for
	\[
	 	e(\tau) 
		\ce
		\begin{cases}
		 X & \text{ if $ \tau= \rho$},
		 \\
		 e_i(\tau) & \text{ if  $\tau \in V_i$}.
		 \end{cases}
	\]
	
	Assume now that~$\C,X\ltdef_e\ltf{\varphi}$ for some~$e \colon V \to \FF(C)$, where~$e(\rho) = X$ by Def.~\ref{def:embedding}\eqref{emb:root}~and $\treeLab(\rho) \subseteq \chi(X)$ by Def.~\ref{def:embedding}\eqref{emb:node}. 
	Hence,  \eqref{critone}~fails. 
	Let~$Y_i \ce e(\sigma_i)$. 
	It~is easy to verify that~$\C,Y_i \ltdef_{e\upharpoonright V_i} \ltf{\psi_i}$. 
	It~remains to note that, by Def.~\ref{def:embedding}\eqref{emb:edge}, $a_i =\treeLab(\rho,\sigma_i) \in \chi\bigl(e(\rho) \cap e(\sigma_i)\bigr)$, i.e.,~$ a_i  \in \chi(X \cap Y_i)$. 
	Thus, \eqref{crittwo}~also fails.
\end{proof}

Using life trees, we can now provide the localized transformation from undefined to true formulas:
\begin{lemma}
	\label{formula.for.Non.embedded.LT}
	Let~$\C,X \nfctdef \varphi$ for some formula~$\varphi \in \glocalK$ and some star-finite pointed simplicial model\/~$(\C,X)$. 
	Then~there is a formula~$\for{X}{\Tree} \in \glocalK$ such that 
	\begin{gather}
	\label{eq:notembeddable}
	\C,X \fctsat \neg \for{X}{\Tree},
	\\
	\label{eq:embeddable}
	\C',X' \fctdef \varphi \quad\Longrightarrow\quad \C',X'  \fctsat \for{X}{\Tree} \quad\text{for all\/ $(\C',X')$}. 
	\end{gather}
\end{lemma}
\begin{proof} 
	Let~$\C=(C,\chi,\ell)$. 
	By~Lemma~\ref{EmbeddingAsDefinability}, $\C,X \nltdef \ltf{\varphi}$.
	We~construct $\for{X}{\Tree}$ by recursion on the~$\depth{(\ltf{\varphi})}$ of \rev{the} life tree $\ltf{\varphi}= \bigl((V,E),\treeLab\bigr)$ with root~$\rho \in V$. 
		
		If~$\treeLab(\rho) \nsubseteq \chi(X)$ meaning that~$a \notin \chi(X)$ for some~$a \in \treeLab(\rho)$, we define
			$\for{X}{\Tree} \ce \bigwedge_{c \in \treeLab(\rho) } c$.
		This~$\for{X}{\Tree}$ satisfies \eqref{eq:notembeddable} because $a$~is false in~$X$ due to~$a \notin \chi(X)$. 
		In~addition, \eqref{eq:embeddable}~holds because~$\C',X' \fctdef \varphi$ for~$\C' = (C',\chi',\ell')$ is equivalent to~$\C',X' \ltdef \ltf{\varphi}$ by Lemma~\ref{EmbeddingAsDefinability}, which implies that~$\treeLab(\rho) \subseteq \chi'(X')$ and, hence,~$\C',X' \fctsat \for{X}{\Tree}$. 
		This~case requires no recursive calls.
		
		If~$\treeLab(\rho) \subseteq \chi(X)$, then, by Lemma~\ref{Non.embedded.Tree}, criterion~\eqref{crittwo} must hold, i.e.,~there is a subformula~$\psi$ of~$\varphi$ and an $a$-edge~$(\rho,\rho_\psi) \in E$ to the root~$\rho_\psi$ of the \mbox{$a$-grafting}~$\ltf{\psi}^a$ where~$\ltf{\psi}=\bigl((V_\psi,E_\psi),\treeLab_\psi\bigr)$ such that~$\C, Y \nltdef \ltf{\psi}$ for any~$Y$ from~\mbox{$\sstar_a(X)\ce \{Y \in 
		\FF(\C) \mid a \in \chi(X \cap Y)\}$}\rev{, 
		which is finite by star-finiteness and~$X \in \sstar_a(X)$. 		
		Indeed, $a=\treeLab(\rho,\rho_\psi) \in\treeLab(\rho) \subseteq \chi(X)= \chi (X \cap X)$.
		Let~$\sstar_a(X) = \{Y_1,\dots,Y_r\}$.}
		By~Lemma~\ref{EmbeddingAsDefinability}, $\C, Y_i \nfctdef \psi$ for~$1 \leq i \leq r$, so we can recursively construct~$\forrk{Y_i}{\psi}$ such that
		\begin{gather}
			\label{eq:notembeddable_proof}
			\C,Y_i \fctsat \neg \forrk{Y_i}{\psi},
			\\
			\label{eq:embeddable_proof}
			\C',X' \fctdef \psi \quad\Longrightarrow\quad \C',X'  \fctsat \forrk{Y_i}{\psi} \quad\text{for all $(\C',X')$}. 
		\end{gather}
		for  each~$1 \leq i \leq r$  because~$\depth{(\ltf{\psi})} < \depth{(\ltf{\varphi})}$. 
		We define a sequence of formulas~$\xi_1,\dots,\xi_r$  as follows: 
		\begin{align*}
			\xi_1 &\ce \forrk{Y_1}{\psi}
			\\
			\xi_{k+1}&\ce 
			\begin{cases}
				\xi_{k} & \text{ if $\C, Y_{k+1} \nfctsat \xi_{k}$;}
				\\
				\xi_{k} \land \forrk{Y_{k+1}}{\psi} & \text{ otherwise, i.e., if $\C, Y_{k+1} \fctsat \xi_{k}$.}
			\end{cases}
		\end{align*}
		By~induction on~$1 \leq k \leq r$ we prove that 
		\begin{gather}
			\label{eq:sub_onefalse}
			\C, Y_{j_k} \fctsat \lnot \xi_k \text{ for some $1 \leq j_k \leq k$};
			\\
			\label{eq:sub_allnottrue}
			\C, Y_l \nfctsat \xi_k \text{ for any $1\leq l \leq k$}. 
		\end{gather}
		For~the \emph{base case~$k=1$} we must choose~$j_1\ce1$, which works  because~$\xi_1 =\forrk{Y_1}{\psi}$,  yielding both~$\C,Y_1 \fctsat \neg \forrk{Y_1}{\psi}$ by~\eqref{eq:notembeddable_proof} for~\eqref{eq:sub_onefalse}~and $\C,Y_1 \nfctsat  \forrk{Y_1}{\psi}$ for~\eqref{eq:sub_allnottrue}. 
		For~the \emph{induction step}, assume \eqref{eq:sub_onefalse}--\eqref{eq:sub_allnottrue}~hold for~$k$ and consider~$\xi_{k+1}$:
		\begin{compactitem}
			\item 
			If~$\xi_{k+1} =\xi_{k}$ because~$\C, Y_{k+1} \nfctsat \xi_{k}$, then we set~$j_{k+1} \ce j_{k}$. 
			By~the~IH\eqref{eq:sub_onefalse},  $\C, Y_{j_{k}} \fctsat \lnot \xi_{k}$, yielding~$\C, Y_{j_{k+1}} \fctsat \lnot \xi_{k+1}$, fulfilling~\eqref{eq:sub_onefalse}. 
			Further, by the~IH\eqref{eq:sub_allnottrue}, $\C, Y_l \nfctsat \xi_{k}$ for any~$1 \leq l \leq k$~and $\C, Y_{k+1} \nfctsat \xi_{k}$ by construction. 
			Hence, $\C, Y_l \nfctsat \xi_{k+1}$ for any~$1 \leq l \leq k+1$ fulfilling~\eqref{eq:sub_allnottrue}.
			\item 
			If,~on the other hand, $\xi_{k+1}=\xi_{k} \land \forrk{Y_{k+1}}{\psi}$ because~$\C, Y_{k+1} \fctsat \xi_{k}$, then we set~$j_{k+1} \ce k+1$.
			By~\eqref{eq:notembeddable_proof},
			$\C, Y_{k+1} \fctsat  \lnot \forrk{Y_{k+1}}{\psi}$, so~$\C, Y_{k+1} \fctsat \lnot (\xi_{k} \land \forrk{Y_{k+1}}{\psi})$
			fulfilling~\eqref{eq:sub_onefalse}. 
			Further, by the~IH\eqref{eq:sub_allnottrue}, $\C, Y_l \nfctsat \xi_{k}$ for any~$1 \leq l \leq k$ and, as already discussed, $\C, Y_{k+1} \nfctsat   \forrk{Y_{k+1}}{\psi}$. 
			Hence, $\C, Y_l \nfctsat \xi_{k} \land \forrk{Y_{k+1}}{\psi}$ for any~$1 \leq l \leq k+1$ fulfilling~\eqref{eq:sub_allnottrue}. 
		\end{compactitem}
		This~completes the induction proof of~\eqref{eq:sub_onefalse}--\eqref{eq:sub_allnottrue}. 
		In~particular, for~$k=r$, we have  that $\xi_r$~is not true in any facet $a$-adjacent to~$X$ and is false, hence, defined in at least one such facet. 
		Thus, $\C, X \fctsat \lnot \M_a \xi_r$,~and  \eqref{eq:notembeddable}~holds for~$\for{X}{\Tree} \ce \M_a \xi_r$.

		It remains to show that \eqref{eq:embeddable}~also holds. 
		Consider any~$(\C',X')$ with~\mbox{$\C'=( C',\chi',\ell')$} such that~$\C',X' \fctdef \varphi$. 
		By~Lemma~\ref{EmbeddingAsDefinability},  $\C',X' \ltdef_{e'} \ltf{\varphi}$ for some~$e'\colon V \to \FF(C')$ with~$e'(\rho)= X'$. Let $e'(\rho_\psi)=Z'$. 
		Since $(\rho,\rho_\psi) \in E$ is an $a$-edge, $a \in \chi'(X' \cap Z')$ by Def.~\ref{def:embedding}\eqref{emb:edge}. 
		It~is easy to see that~$\C', Z' \ltdef_{e'\upharpoonright V_{\psi}} \ltf{\psi}$, thus,  $\C', Z' \fctdef \psi$ by Lemma~\ref{EmbeddingAsDefinability}. 
		By~\eqref{eq:embeddable_proof}, $\C',Z' \fctsat \forrk{Y_i}{\psi}$ for all~$1\leq i \leq r$. 
		Given that $\xi_r$~is the conjunction  of some of~$\forrk{Y_i}{\psi}$'s, we conclude that~$\C',Z' \fctsat \xi_r$.  
		Consequently, $\C', X' \fctsat \M_a \xi_r$, i.e.,~$\C', X' \fctsat \for{X}{\Tree}$, which completes the proof of~\eqref{eq:embeddable}.
\end{proof}

	\begin{figure}[t]
	\centering
\scalebox{.8}{
		\begin{tabular}{c@{\qquad\qquad}cc}
			\begin{tikzpicture}
				\fill[fill=gray!25!white] (0,0) -- (-1,-1.5)  -- (1,-1.5) -- cycle;
				
				\node[a dot] (a0) at (0,0) {$1_a$};
				\node[a dot] (a3) at (-2,-3) {$0_a$};
				\node[b dot] (b3) at (-1,-1.5) {$0_b$};
				\node[c dot] (c3) at (1,-1.5) {$1_c$};
				\node[d dot] (d4) at (2,-3) {$0_d$};
				\node[e dot] (e2) at (0,-3) {$1_e$};
				\node(c) at (2,-2.3) {$Y_2$};
				\node(c) at (0, -2.3) {$Y_3$};
				\node(c) at (-1.1, -2.7) {$Y_4$};
				\node(c) at (0, -1) {$X$};
				
				\draw[-] (a0) -- (b3)--(c3)--(a0);
				\draw[-] (c3) --(e2)--(a3);
				\draw[-] (c3) -- (d4);
			\end{tikzpicture}
			&
			\begin{tikzpicture}
				\node (nothing) at (1,4.5) { };
				\node (L1r) at (1,1.5) {};
				\node (L2m) at (1,3) {$\{d \}$};
				\node (L3m) at (,4.5) {$\{a,d \}$};

				\draw (L2m)--(L3m)  node at (.7,3.7) {$d$};
			\end{tikzpicture}
			&
			\begin{tikzpicture}
				\node (nothing) at (-2,4.5) { };
				\node (L1r) at (1,1.5) {$\{b,c\}$};
				\node (L2m) at (-.5,3) {$\{c,d \}$};
				\node (L2r) at (2.5,3) {$\{c,e\}$};
				\node (L3m) at (-.5,4.5) {$\{a,d \}$};
				\node (L3r) at (2.5,4.5) {$\{a,e\}$};

				\draw (L1r)--(L2m) node at (-0.2,2.2) {$c$};
				\draw (L1r)--(L2r) node at  (2.2,2.2) {$c$};
				\draw (L2m)--(L3m)  node at (-.8,3.7) {$d$};
				\draw (L2r)--(L3r) node at (2.3,3.7) {$e$};
			\end{tikzpicture}
		\end{tabular}
}
		\caption{From right to left: life tree~$\ltf{\varphi}$ of formula~$\varphi =  p_b \land \M_c \M_d p_a \land \M_c \M_e \neg p_a$; life tree~$\ltf{\psi}$ of its subformula~$\psi =\M_d p_a$; and
		simplicial model~$(\C,X)$ such that~$\C, X \nfctdef \varphi$.}
		\label{fig.rojo2}
	\end{figure}

\begin{example} 
\label{example.simplified}
	\rev{Consider the pointed simplicial model}~$(\C,X)$ in Fig.~\ref{fig.rojo2}~(left). 	
	Let~$\varphi =  p_b\land \M_c \M_d p_a \land \M_c \M_e \neg p_a$. 
	Its~life tree~$\ltf{\varphi}=\bigl((V,E),\treeLab\bigr)$ with root~$\rho$ is depicted in Fig.~\ref{fig.rojo2}~(right). 
	We~will show that~$\C,X \nfctdef \varphi$ by using Lemmas~\ref{EmbeddingAsDefinability}~and~\ref{Non.embedded.Tree}, while simultaneously  constructing~$\forrk{X}{\varphi}$ in two different ways.
	All agents from~$\treeLab(\rho)=\{b,c\}$ of its root~$\rho$ are present in~$\chi(X) = \{a,b,c\}$. 
	Thus, criterion~\eqref{critone} of Lemma~\ref{Non.embedded.Tree} fails. 
	But~consider the left subtree~$\ltf{\psi}^c$ of~$\ltf{\varphi}$, which corresponds to  subformula~$\psi = \M_d p_a$ of~$\varphi$. 
	Life tree~$\ltf{\psi}$, see Fig.~\ref{fig.rojo2}~(middle), cannot be embedded into any of the facets $c$-adjacent to~$X$: indeed,  label~$\{d\}$ of its root~$\rho_\psi$ is  included in neither~$\chi(X)= \{a,b,c\}$~nor $\chi(Y_3)=\{c,e\}$, resulting in~$\forrk{X}{\psi} =\forrk{Y_3}{\psi} = d$, according to the construction in the proof of Lemma~\ref{formula.for.Non.embedded.LT}. 
	For~the remaining, third $c$-adjacent facet~$Y_2$ with~$\chi(Y_2)= \{c,d\}$,  the label~$\{d\}$ of~$\rho_\psi$ is included, however, $\ltf{\psi}$~has a $d$-edge from its root to a node labeled~$\{a,d\}$, which corresponds to subformula~$p_a$, and $a$~is absent in the only facet $d$-adjacent to~$Y_2$, which is $Y_2$~itself. 
	Therefore, the same construction yields~$\forrk{Y_2}{\psi} = \M_d a$. 
	
	How~$\forrk{X}{\varphi}$~is constructed from~$\forrk{X}{\psi}$,  $\forrk{Y_2}{\psi}$,~and~$\forrk{Y_3}{\psi}$ depends on the chosen ordering among these three $c$-adjacent facets: 

\noindent	
\strut\hfill
	\begin{tabular}{l|l}
	 \multicolumn{1}{c|}{Ordering $Y_2, X, Y_3$} & \multicolumn{1}{c}{Ordering $X, Y_2, Y_3$}
	\\
	\hline
	$\xi^{\mathstrut}_1 = \forrk{Y_2}{\psi} = \M_d a$ & $\xi'_1 = \forrk{X}{\psi}= d$
	\\
	$\xi_2 = \xi_1 = \M_d a$\quad as\quad$\C,X \nfctsat \M_d a$ & 
	$\xi'_2 = \xi'_1 \land  \forrk{Y_2}{\psi} = d \land \M_d a$\quad as\quad$\C,Y_2 \fctsat d$
	\\
	$\xi_3 = \xi_2 = \M_d a$\quad  as\quad$\C,Y_3 \nfctsat \M_d a$ & $\xi'_3 = \xi'_2 = d \land \M_d a$\quad as\quad$\C,Y_3 \nfctsat d \land \M_d a$
	\end{tabular}
	\hfill\strut
Accordingly, the two orderings provide two alternative formulas~$\M_c\M_d a$~or $\M_c(d \land \M_da)$ for the role of~$\forrk{X}{\varphi}$. 
Either of them satisfies~\eqref{eq:notembeddable}--\eqref{eq:embeddable}.
\end{example}

\begin{theorem}[Modal equivalence implies bisimilarity] 
\label{hehelastone}
For arbitrary star-finite  pointed simplicial models\/~$(\C,X)$~and\/ $(\C',X')$,\\
\strut\hfill$(\C,X) \equiv_{\glocalK} (\C',X') \quad \Longrightarrow\quad(\C,X) \bisim (\C',X')$.\hfill\strut
\end{theorem}
\begin{proof}
	Let~$(\C,X)\equiv_{\glocalK} (\C',X')$ for~$\C = (C, \chi, \ell)$~and $\C' = (C', \chi', \ell')$.
	We~define a binary relation~$\B\subseteq \FF(C) \times \FF(C')$ so that~$Y\B Y'$ if{f} $(\C,Y) \equiv_{\glocalK} (\C',Y')$. 
	In~particular,~$X \B X'$. 
	Towards a contradiction, assume that  $\B$~is not a bisimulation. 
	We~can prove without difficulty that \textbf{atoms} is satisfied. 
	Thus, either \textbf{forth} or \textbf{back} fails. 
	The~two cases  are symmetric, and we consider only the former.   
	Assume  that~$Z \B Z'$, but for some~$a \in A$~and some~$ Y \in \FF(C)$ with~$a \in \chi(Z \cap Y)$, there is no~$Y' \in \FF(C')$ such that~$a \in \chi'(Z' \cap Y')$~and~$Y \B Y'$. 
	Let~$\sstar_a(Z') = \{Y'_1,\dots,Y'_n\}$ in $\C'$
	(this~set is finite because $\C'$~is star-finite). 
	By~construction of~$\B$,  we have~$(\C,Y) \not\equiv_{\glocalK} (\C',Y'_i)$ for any~$1 \leq i\leq n$. 
	In~other words, there exist formulas~$\varphi_1,\dots,\varphi_n$ such that for each~$1 \leq i \leq n$ one of the following three statements holds:
	\begin{align}
		\label{eq:truenottrue}
		\C, Y &\fctsat \varphi_i  &\text{but }\qquad \C', Y'_i &\nfctsat \varphi_i,
		\\
		\label{eq:falsenotfalse}
		\C, Y &\fctsat \lnot\varphi_i  &\text{but }\qquad \C', Y'_i &\nfctsat \lnot\varphi_i,
		\\
		\label{eq:defnotdef}
		\C, Y &\nfctdef \varphi_i \qquad &\text{but }\qquad \C', Y'_i &\fctdef \varphi_i.
	\end{align}
	We transform these distinguishing formulas into~$\psi_1,\dots,\psi_n$ as follows:
	\[
	\psi_i \ce 
	\begin{cases}
		\varphi_i & \text{if \eqref{eq:truenottrue} holds for $\varphi_i$,}
		\\
		\lnot \varphi_i & \text{if \eqref{eq:falsenotfalse} holds for $\varphi_i$,}
		\\
		\lnot \forrk{Y}{{(\varphi_i)}} & \text{if \eqref{eq:defnotdef} holds for $\varphi_i$.}
	\end{cases}
	\]
	Note~that now $\C, Y \fctsat \psi_i$ and $\C', Y'_i \nfctsat \psi_i$ for all~$1 \leq i \leq n$. 
	Indeed, for the first two clauses, it follows directly from~\eqref{eq:truenottrue}--\eqref{eq:falsenotfalse}, while for the last clause this is a consequence of Lemma~\ref{formula.for.Non.embedded.LT}.
	In~other words, we have replaced \rev{the} initial distinguishing formulas~$\varphi_i$ with  \rev{the} distinguishing formulas~$\psi_i$ that are all true in~$(\C,Y)$. 
	We~will show that~$
	\M_a \Et_{i=1}^n \psi_i$ distinguishes~$Z$ from~$Z'$.
	This contradicts our assumption that~$(\C,Z)\equiv_{\glocalK}(\C',Z')$ and, thus, proves that $\B$~is a bisimulation.
	
	Since~$\C,Y\fctsat\psi_i$ for all~$1 \leq i \leq n$, also~$\C,Y\fctsat\Et_{i=1}^n\psi_i$. 
	Given~$a \in \chi(Z \inter Y)$, we obtain~$\C,Z\fctsat \M_a\Et_{i=1}^n\psi_i$. 
	Since~$\C',Y'_i\nfctsat\psi_i$, it follows that~$\C',Y'_i\nfctsat\Et_{i=1}^n\psi_i$ for any~$1 \leq i \leq n$. 
	Thus, $\C',Z'\nfctsat \M_a\Et_{i=1}^n\psi_i$.
\end{proof}

\begin{theorem}[Hennessy--Milner property]
	\label{Hennessy--Milner}
For arbitrary star-finite  pointed simplicial models\/~$(\C,X)$~and\/ $(\C',X')$,\\
\[
(\C,X)\equiv_{\glocalK}(\C',X')
\quad\Longleftrightarrow\quad
(\C,X)\bisim(\C',X').
\] 
\end{theorem}
\begin{proof}
	This follows from Theorems~\ref{BisimImpliesEquiv}~and~\ref{hehelastone}.
\end{proof}

\section{Life Bisimulation for Partial Epistemic Models} \label{sec.kripke}

\rev{The~main purpose of this section is to provide a bridge to a more familiar formalism for those who are not yet fluent in simplicial semantics. 
In~particular,  we compare simplicial models to Kripke models, including our results for bisimulation. 
Accordingly, we omit proofs in this section because they are obtained via the categorical equivalence with simplicial models.} 
For~the model correspondence in  language~$\localK$ we recall~\cite{vDitKuz22arXiv}, which only  requires a minor modification to be extended to  language~$\glocalK$. 
We~continue by defining life bisimulation and will contrast it with standard  bisimulation. 
Note that languages~$\glocalK$~and~$\localK$ remain the same throughout, we merely interpret them on Kripke models now.

\begin{definition}
A~binary relation~$\sim$ on \rev{a set}~$S$ is called a \emph{partial equivalence relation}~if{f} \rev{$\sim$~is an equivalence relation on some subset~$S' \subseteq S$.}   
\end{definition}
In~particular, partial equivalence relations are transitive and symmetric and induce a partition of~$S'$. 
Note also that $S'$~is uniquely determined by~$\sim$: $S' = \{s \in S \mid s \sim s\}$.

\begin{definition}
Consider a pair~$(S,\sim)$ where $S$~is the \emph{domain} of (\emph{global})~\emph{states}~and $\sim \colon A \to 2^{S \times S}$ maps each agent~$a \in A$ to a partial equivalence relation~$\sim_a$ on~$S$. 
Let~$S_a$~be the  subset of~$S$ such that $\sim_a$~is an equivalence relation on~$S_a$, i.e.,~the set of states where $a$~is alive. 
Let~$[s]_a \ce \{ t \in S \mid s \sim_a t \}$ denote  equivalence classes of~$\sim_a$ on~$S_a$. 
$(S,\sim)$~is \emph{image-finite} if{f} $[s]_a$~is finite for all~$a \in A$~and $s\in S_a$. 
Given~$s \in S$, \rev{the} set~$A_s \ce \{a \in A \mid s \in S_a\}$ contains the agents that are alive in state~$s$.
Relation~$\sim_a$ is \emph{proper} if{f} for all distinct~$s,t \in S$ there is a\rev{n agent}~$b \in A_s$ such that~$s \not\sim_b t$. 

For~such a pair~$(S,\sim)$, a triple~$\model = (S,\sim,L)$ is  a \emph{partial epistemic model}  if{f}  all~$\sim_a$ are proper and a \emph{valuation function}  $L \colon S \to 2^{P}$ satisfies: for all~$a \in A$, $p_a \in P_a$, $s,t \in S_a$,~and $u \in S$,
\begin{compactitem}
\item if~$s \sim_a t$\rev{,} then~$p_a \in L(s)$ if{f} $p_a \in L(t)$, and 
\item $a \in L(u)$ if{f} $a \in A_u$. 
\end{compactitem}
$(\model,s)$~for~$s \in S$ is a \emph{pointed partial epistemic model} (we often omit~`pointed').  
\end{definition}

As in Def.~\ref{lanKa.fct.defsat}, to interpret $\varphi\in \glocalK$ in a global state~$s$ of a partial epistemic model~$\model$,  by induction on the structure of~$\varphi$,   we define~$\fctdef$ to determine whether $\varphi$~is defined~and $\fctsat$~to determine its truth value when defined. 
\begin{definition} 
Given~a partial epistemic model~$\model = (S,\sim,L)$,  for all~$s \in S$ we define~$\fctdef$~and~$\fctsat$ by induction on~$\varphi\in\glocalK$:\\
$\begin{array}{lcl}
\model,s \fctdef a & &\text{always}; 
\\
\model,s \fctdef p_a & \text{if{f}} & s \in S_a; \\
\model,s \fctdef \neg\varphi & \text{if{f}} & \model,s \fctdef \varphi; 
\\
\model,s \fctdef \varphi\et\psi & \text{if{f}} & \model,s \fctdef \varphi \text{ and } \model,s \fctdef \psi; \\
\model,s \fctdef \M_a \varphi & \text{if{f}} & \model,t \fctdef \varphi \text{ for some } t \in S \text{ such that } t \sim_a s. \\[.5ex]
\model,s \fctsat a & \text{if{f}} & s \in S_a; \\
\model,s \fctsat p_a & \text{if{f}} & s \in S_a \text{ and }  p_a \in L(s);\\
\model,s \fctsat \neg\varphi & \text{if{f}} & \model,s \fctdef \varphi \text{ and } \model,s \nfctsat \varphi;\\
\model,s \fctsat \varphi\et\psi & \text{if{f}} & \model,s \fctsat \varphi \text{ and } \model,s \fctsat \psi;\\
\model,s \fctsat \M_a \varphi & \text{if{f}} & \model,t \fctsat \varphi \text{ for some } t \in S \text{ such that } t \sim_a s.
\end{array}$\\
Formula~$\varphi$ is \emph{valid} if{f}  $\model,s \fctdef \varphi$ implies $\model,s \fctsat \varphi$ for all~$(\model,s)$. 
The~\emph{denotation} of~$\varphi$ in~$\model$ is defined as \rev{the} set~$\I{\varphi}_\model\ce\{ s \in S \mid \model,s \fctsat \varphi \}$.

Partial epistemic models~$(\model,s)$~and $(\model',s')$ are \emph{modally equivalent}, 
denoted~$(\model,s)\equiv^\loc(\model',s')$, if{f} for 
all~$\varphi\in\glocalK$, 
$\model,s \fctdef\varphi \Longleftrightarrow\model',s' \fctdef \varphi$,  
$\model,s \fctsat \varphi\Longleftrightarrow\model',s' \fctsat \varphi$,~and 
$\model,s \fctsat\neg \varphi\Longleftrightarrow\model',s' \fctsat \neg \varphi$.
\end{definition}

A~consequence of our semantics is that local \rev{atoms}~$p_a$ for  agent~$a$ may be assigned to states~$s$ outside of~$S_a$, i.e.,~where $a$~is dead. 
Such~\rev{atoms}~$p_a$ are undefined in~$s$ (and~\rev{atom}~$a$ is~false)  whether $p_a \in L(s)$~or $p_a \notin L(s)$. 
Partial epistemic models, therefore, contain superfluous information. An alternative knowledge representation would make them truly partial, albeit at the expense of comparing life bisimulation with standard bisimulation.

We~now recall the correspondence between impure simplicial models and partial epistemic models from~\cite{vDitKuz22arXiv}, which  generalizes \rev{the equivalence of categories from~\cite{GoubaultLR21} between pure simplicial models and local proper Kripke models where all relations are equivalence relations.}

\begin{definition}
Operation~$\sigma$ (for~\emph{S}implicial) maps each partial epistemic model~$\model= (S,\sim,L)$ to a simplicial model~$\sigma(\model)= (C,\chi,\ell)$ as follows:
\begin{compactitem}
	\item 
	vertices  are  pairs $\bigl([s]_a,a\bigr)$ for all~$s \in S$~and    $a \in A_s$;
	\item 
	$C$~consists of simplexes~$ \bigl\{ \bigl([s]_a,a\bigr) \,\big|\, a \in B \bigr\}$ for all~$s \in S$~and   $\emptyset\neq B \subseteq A_s$; 
	\item 
$\chi\Bigl(\bigl([s]_a,a\bigr)\Bigr) \ce a$ for each vertex~$\bigl([s]_a,a\bigr)$;
	\item  
$\ell\Bigl(\bigl([s]_a,a\bigr)\Bigr) \ce  P_a \cap L(s)$ for each vertex~$\bigl([s]_a,a\bigr)$.
\end{compactitem}
We~let $\sigma(s)$~denote  \rev{the} facet $\left\{ \bigl([s]_a,a\bigr) \mid a \in A_s\right\}$. 

Operation~$\kappa$ (for~\emph{K}ripke)  maps each  simplicial model~$\C= (C,\chi,\ell)$ to a partial epistemic model~$\kappa(\C) = (S,\sim,L)$ as follows:
 \begin{compactitem}
	\item 
	$S \ce \FF(C)$ consists of  facets~$X \in \FF(C)$;
	\item 
$X \sim_a Y $ if{f} $a \in \chi(X \cap Y)$ for any agent~$a$ and global states~$X$~and~$Y$;
	\item 
$L(X) \ce \ell(X) \cup  \chi(X)$ for any global state~$X$. 
\end{compactitem} 

As~$\sigma$~maps each state~$s$ in~$\model$ to a facet~$\sigma(s)$ in~$\sigma(\model)$~and $\kappa$~maps each facet~$X$ in~$\C$ to a state~$X$ in~$\kappa(\C)$, these maps are also between structures~$(\model,s)$ respectively~$(\C,X)$: 
we let $\sigma(\model,s) \ce \bigl(\sigma(\model),\sigma(s)\bigr)$~and $\kappa(\C,X) \ce \bigl(\kappa(\C), X\bigr)$.
\end{definition}
We recall from~\cite{vDitKuz22arXiv} that for all~$\varphi\in\localK$, 
$\model,s\fctdef\varphi \Longleftrightarrow\sigma(\model,s) \fctdef\varphi$~and  
$\C,X\fctdef\varphi\Longleftrightarrow\kappa(\C,X)\fctdef \varphi$.
It is straightforward to extend this to~$\glocalK$.
\begin{proposition}
\label{prop.corr2} 
Let~$\varphi \in \glocalK$.
\begin{compactitem}
\item For~all pointed partial epistemic models\/~$(\model,s)$: \\ 
$\model,s\fctdef \varphi \Longleftrightarrow\sigma(\model,s) \fctdef \varphi$\quad and\quad$\model,s\fctsat \varphi \Longleftrightarrow\sigma(\model,s) \fctsat \varphi$.
\item 
For~all pointed simplicial models\/~$(\C,X)$: \\ 
$\C,X \fctdef\varphi \Longleftrightarrow\kappa(\C,X) \fctdef \varphi$\quad and\quad$\C,X \fctsat \varphi \Longleftrightarrow\kappa(\C,X) \fctsat \varphi$.
\end{compactitem}
\end{proposition}

Figure~\ref{fig.motivating} contains examples of corresponding simplicial models and partial epistemic models. 
Similarly,  simplicial models from Example~\ref{example.bisim} correspond to partial epistemic models  from Example~\ref{example.bisim2} below.

We~now define \emph{life bisimulation} for partial epistemic models, show how it corresponds to impure simplicial models, and how it is different from the standard notion of bisimulation for Kripke models.

\begin{definition}
A~\emph{life bisimulation} between partial epistemic models~\mbox{$\model = (S,\sim,L)$}~and $\model' = (S',\sim',L')$, notation~$\Z: \model\bisim \model'$,~or \mbox{$\Z: (\model,s)\bisim (\model',s')$} given~$s\Z s'$, is a non-empty binary relation~$\Z \subseteq S \times S'$ such that for all~$s \in S$~and $s' \in S'$ with~$s\Z s'$ the following three conditions are satisfied:
\begin{compactitem}
\item \textbf{Atoms}: 
$L(s) \cap A = L'(s') \cap A$~and, additionally, $L(s) \cap P_a = L'(s') \cap P_a$ for each~$a \in A_s$ (note~that here $A_s = L(s) \cap A = L'(s') \cap A = A_{s'}$).
\item \textbf{Forth}: 
for all~$a \in A_s$, for all~$t \sim_a s$, there is a~$t'\sim'_a s'$ such that~$t\Z t'$.
\item \textbf{Back}: 
for all~$a \in A_{s'}$, for all~$t'\sim'_as'$, there is a~$t \sim_a s$ such that~$t\Z t'$.
\end{compactitem}
Life~bisimulation~$\Z$ is \emph{total} \rev{if{f}} the domain and codomain of~$\Z$ are~$S$ \rev{respectively}~$S'$. 
When~$\Z$~is omitted, a life bisimulation must exist.
\end{definition}

A~\emph{standard~bisimulation}~\cite{blackburnetal:2001}, notation~$\bisim^\sta$, can be obtained from a life bisimulation by replacing the requirements~$a \in A_s$~(twice)~and $a \in A_{s'}$ in the above definition with~$a \in A$. 
Standard bisimilarity, therefore, implies life bisimilarity: 
if~$(\model,s)\bisim^\sta (\model',s')$, then~$(\model,s)\bisim (\model',s')$. 
On~the class of multiagent \rev{$\SFive$}~models, where all partial equivalence relations are equivalence relations, $(\model,s)\bisim^\sta (\model',s')$ if{f} $(\model,s)\bisim (\model',s')$. 
In~the states of \rev{$\SFive$}~models, all formulas are defined, and the semantics becomes two-valued. Other than that, comparing life and standard bisimulations is a bit like comparing apples to onions, as our semantics is three-valued. See also the example below.

\begin{example} \label{example.bisim2}
\rev{Partial epistemic models~$\model$, $\model'$,~and~$\model''$ below correspond to the simplicial models of Example~\ref{example.bisim} and are life bisimilar. 
The~states are named with the values of the local \rev{atoms} of the live agents. 
As~in Example~\ref{example.bisim}, a total bisimulation~$\Z$  between~$\model$~and~$\model'$  requires that~$X\Z X'$~and~$X\Z Z'$, while $\Z$~between~$\model$~and~$\model''$ must have~$Y\Z Y''$~and~$Y\Z Z''$. }
\begin{center}
\scalebox{0.8}{
\begin{tikzpicture}
\node[Kstate] (010) at (.5,0) {$0_a1_b$};
\node[Kstate] (001) at (3.5,0) {$0_a0_b1_c$};
\node (x) at (.5,-.6) {$X$};
\node (y) at (3.5,-.6) {$Y$};
\node (b) at (1.75,-1.5) {\color{white} $0_a1_b0_c$};
\draw[-] (010) -- node[above] {$a$} (001);
\end{tikzpicture}
\qquad
\begin{tikzpicture}
\node[Kstate] (010) at (.5,0) {$0_a1_b$};
\node[Kstate] (001) at (3.5,0) {$0_a0_b1_c$};
\node[Kstate] (b) at (1.9,-1.5) {$0_a1_b$};
\node (xa) at (0.2,-.6) {$X'$};
\node (ya) at (3.7,-.6) {$Y'$};
\node (za) at (2.8,-1.5) {$Z'$};
\draw[-] (010) -- node[above] {$a$} (001);
\draw (010) --(b)  node at (1,-.8) {$a$} ;
\draw[-] (b) -- (001)  node at (3,-.8) {$a$};
\end{tikzpicture}
\qquad
\begin{tikzpicture}
\node[Kstate] (010) at (.5,0) {$0_a1_b$};
\node[Kstate] (001) at (3.5,0) {$0_a0_b1_c$};
\node[Kstate] (b) at (1.9,-1.5) {$0_a0_b1_c$};
\node (xa) at (0.2,-.6) {$X''$};
\node (ya) at (3.8,-.6) {$Y''$};
\node (za) at (2.8,-1.5) {$Z''$};
\draw[-] (010) -- node[above] {$a$} (001);
\draw (010) --(b)  node at (1,-.8) {$a$} ;
\draw[-] (b) -- (001)  node at (3,-.8) {$a,b$};
\end{tikzpicture}
}
\end{center}
Let~$p_c$~be false in~$X$, false in~$X'$, true in~$Z'$,~and true in~$X''$. 
Then~no two of these models are standard bisimilar. 
$(i)$~Neither~$\model$~and~$\model'$, nor~$\model$~and~$\model''$ are standard bisimilar because $\model'$~and~$\model''$~contain a state~$s$ with~$L(s) = \{p_b,p_c\}$ whereas $\model$~does not contain such a state. 
$(ii)$~Whereas $\model'$~and~$\model''$~are not standard bisimilar because $p_c \in L(s)$ for all states in~$\model''$ but not in~$\model'$.
\end{example}

We~will now show that the~$\sigma$~and~$\kappa$~transformations preserve bisimilarity, and also the Hennessy--Milner property for life bisimulation between partial epistemic models. 
There~are actually two ways to go about this: a direct proof using life trees and embeddings in partial epistemic models and an indirect proof using the results already obtained for simplicial models. 
We~will do the latter as that proof is easy.

\begin{proposition}
\label{prop:bisimEquivMods}
\begin{compactenum}
\item\label{itemi} 
$\model \bisim \model'$ implies $\sigma(\model) \bisim \sigma(\model')$ for any partial epistemic models~$\model$~and~$\model'$. 
\item\label{itemii} 
$\C \bisim \C'$ implies $\kappa(\C) \bisim \kappa(\C')$ for any simplicial models~$\C$~and~$\C'$. 
\end{compactenum}
\end{proposition}

\begin{corollary} \label{parbis_image_to_Model} 
$\model \bisim \kappa\bigl(\sigma(\model)\bigr)$~and $\C \bisim \sigma\bigl(\kappa(\C)\bigr)$.
\end{corollary}
Consequently, $\model \bisim \model'$ if{f} $\kappa\bigl(\sigma(\model)\bigr) \bisim \kappa\bigl(\sigma(\model')\bigr)$~and, by the same token, $\C \bisim \C'$ if{f} $\sigma\bigl(\kappa(\C)\bigr) \bisim \sigma\bigl(\kappa(\C')\bigr)$. 
Same~for \rev{the pointed versions}. 

\begin{theorem}[Hennessy--Milner property]
\label{Hennessy--Milner2}
For~any image-finite partial epistemic models\/~$(\model,s)$~and\/~$(\model',s')$
\[
(\model,s)\bisim(\model',s') \quad \Longleftrightarrow\quad (\model,s)\equiv^\loc(\model',s').
\] 
\end{theorem}

\section{Conclusions and Future Work}
\label{conclusion}
\rev{In~this paper, we concentrate on the question of what is a natural notion of bisimulation for impure simplicial complexes. 
The~notion of bisimulation we have defined is indeed  natural for the following two reasons: 
(I)~On~the categorically equivalent structures, namely partial Kripke frames, this is the abstract notion of bisimulation coming from coalgebraic many-valued logic (cf.~\cite{BilkovaDostal2016}) because  partial Kripke frames can be represented by a well-behaved set endofunctor.
(II)~It~has the expected structural \textbf{forth}-and-\textbf{back} conditions, which are easy to check on finite structures, and it captures our intuition on process or behavioral equivalence in this case.

For~this structurally natural notion of bisimulation, we have shown that the local language is invariant but fails to be sufficiently expressive (the Hennessy--Milner property fails). 
Similar situation has been encountered, e.g.,~in~\cite{{Marti2014-MARAHP-6}} or in~\cite{BilkovaDostal2016}, for example, in case of most G\"odel modal logics. 
We~have demonstrated how to enhance the language with global atoms to ensure sufficient expressivity and proved the Hennessy--Milner property for the extended language. 
Unlike the two-valued case, the latter turned out to be quite non-trivial. 
The~reason for that is the lack of symmetry among our three values. 
If~two boolean-valued structures are not modally equivalent, the distinguishing formula is either true or false in the first structure, but, if the other distinguishing truth value is desired, it is sufficient to take the negation. 
In~our three-valued logic, there is an additional case of the distinguishing formula being undefined in the first structure. 
The~transformation of it into a defined (true) formula is so non-trivial that it cannot even be done in a structure-independent way.
\smallskip

\noindent
\textbf{Future work} 
For~the local language, one can ask for a matching notion of model equivalence, for example, in the spirit of logic-induced bisimulations of~\cite{GrootHansenKurz2020}. 
For~the fully abstract coalgebraic treatment of the logics in this paper as three-valued coalgebraic logics, as well as for investigating which three-valued  first-order fragment  corresponds to modal formulas in a sense of Van~Benthem Theorem, essential insights need to be developed first. 
While~the propositional logic~PWK does not behave well in terms of abstract algebraic logic and only relatively recently was investigated in this context~\cite{PaoliP21,BonzioGPP17}, the first-order expansion of~PWK matching our interpretation of modalities has not been explored. 
The~three-element weak Kleene algebra underlying the semantics is not a lattice, which impedes the use of the usual methods,  e.g.,~those employed in~\cite{BilkovaDostal2016}.

Other~directions for further research were proposed by the anonymous reviewers. 
Among~them was the question of the complexity of checking whether two structures are bisimilar. 
We~thank them for their comments and suggestions.}

\providecommand{\ditmarsch}[1]{}\providecommand{\hoek}[1]{}

\end{document}